\documentclass[journal]{IEEEtran} 
\usepackage{cite}  
\usepackage{times}
\usepackage{amsmath,amsfonts,amssymb,mathtools,mathrsfs}

\usepackage{epsfig,amstext,xspace}

\usepackage{amsthm}

\usepackage{graphicx} \usepackage{caption} \usepackage{subcaption}

\usepackage{algpseudocode}
\usepackage{multirow}
\usepackage[table]{xcolor}

\newtheorem{innercustomthm}{Theorem}
\newenvironment{customthm}[1]
  {\renewcommand\theinnercustomthm{#1}\innercustomthm}
  {\endinnercustomthm}

\newtheorem{innercustomcor}{Corollary}
\newenvironment{customcor}[1]
  {\renewcommand\theinnercustomcor{#1}\innercustomcor}
  {\endinnercustomthm}
  
\newtheorem{innercustomlemma}{Lemma}
\newenvironment{customlemma}[1]
  {\renewcommand\theinnercustomlemma{#1}\innercustomlemma}
  {\endinnercustomlemma}

\usepackage{algorithm}
\usepackage{array}
\newtheorem{claim}{Claim}

{\theoremstyle{remark} }
{\theoremstyle{definition}

}

\def\RR{\mathbb R}

\def\cG{\mathcal G}

\def\cH{\mathcal H}

\def\cN{\mathcal N}

\newcommand{\raf}[1]{(\ref{#1})}

\newcommand{\cP}{\ensuremath{\mathcal{P}}}

\title{Quantifying Inefficiency of Fair Cost-Sharing Mechanisms for Sharing Economy}

\author{
Chi-Kin Chau and Khaled Elbassioni
\thanks{C.-K. Chau and K. Elbassioni are with Masdar Institute (e-mail: \{ckchau, kelbassioni\}@masdar.ac.ae).}
\thanks{Abridged version of this paper appears in IEEE Transactions on Control of Network Systems (DOI:10.1109/TCNS.2017.2763747).
}
}

\newif\ifsupplementary
\supplementaryfalse
\begin{document}

\maketitle

\begin{abstract}
Sharing economy is a distributed peer-to-peer economic paradigm, which gives rise to a variety of social interactions for economic purposes. One fundamental distributed decision-making process is coalition formation for sharing certain replaceable resources collaboratively, for example, sharing hotel rooms among travelers, sharing taxi-rides among passengers, and sharing regular passes among users. Motivated by the applications of sharing economy, this paper studies a coalition formation game subject to the capacity of $K$ participants per coalition. The participants in each coalition are supposed to split the associated cost according to a given cost-sharing mechanism. A stable coalition structure is established when no group of participants can opt out to form another coalition that leads to lower individual payments. We quantify the inefficiency of distributed decision-making processes under a cost-sharing mechanism by the strong price of anarchy (SPoA), comparing a worst-case stable coalition structure and a social optimum. In particular, we derive SPoA for common fair cost-sharing mechanisms (e.g., equal-split, proportional-split, egalitarian and Nash bargaining solutions of bargaining games, and usage based cost-sharing). We show that the SPoA for equal-split, proportional-split, and usage based cost-sharing (under certain conditions) is $\Theta(\log K)$, whereas the one for egalitarian and Nash bargaining solutions is $O(\sqrt{K} \log K)$. Therefore, distributed decision-making processes under common fair cost-sharing mechanisms induce only moderate inefficiency.


\end{abstract}

\pagenumbering{arabic}

\section{Introduction} \label{sec:intro}

The rise of  ``sharing economy'' \cite{economist13} has created a new paradigm of social and economic networks, which promotes distributed peer-to-peer interactions and bypasses traditional centralized hierarchal service providers and intermediaries. Sharing economy is often facilitated by the advent of pervasive information technology platforms, especially by mobile computing and digital social platforms. One fundamental distributed decision-making process in sharing economy is coalition formation for sharing resources and facilities with excess capacity among users collaboratively and efficiently. We highlight some examples of sharing economy as follows:
\begin{enumerate}

\item {\em Hotel Room Sharing}: Travelers may share hotel rooms with other fellow travelers, because multiple-occupancy rooms are more economical. The sharing processes are achieved through private arrangements among travelers.

\item {\em Taxi-ride Sharing}: Commuters may share taxi-rides because of lower taxi fares, despite that the taxicabs may take a detour to pick-up or drop-off other passengers.

\item {\em Pass Sharing}: Certain anonymous regular passes are validated within a certain fixed period of time (e.g., accesses to public transportation or public facilities), which may be shareable among multiple holders, if they do not overlap in their usage times. Note that this example also applies to co-owning or co-leasing physical properties (e.g., houses, cars and parking lots).

\end{enumerate}

In this paper, we study a class of distributed decision-making processes for sharing economy. In particular, we consider the problem with a set of participants forming coalitions to share certain replaceable resources from a large pool of available resources (e.g., hotel rooms, taxicabs, regular passes), such that any subset of participants can always form a coalition using separate resources, independent from other coalitions. We formulate a coalition formation game, wherein participants form arbitrary coalitions of their own accord to share the associated cost, subject to a constraint on the maximum number of participants per coalition. 

There are two main aspects investigated in this paper:
\begin{enumerate}

\item {\bf Inefficiency of Distributed Decision-Making}:
Since there is a capacity per coalition such that not all participants can form a single coalition, there will exist potentially multiple coalitions and the self-interested participants will opt for the lower payments.
Distributed self-interested behavior often gives rise to outcomes that deviate from a social optimum. A critical question is related to the inefficiency of distributed decision-making processes. We quantify the inefficiency of distributed decision-making processes by the {\em Strong Price of Anarchy} (SPoA), a common metric in Algorithmic Game Theory that compares the worst-case ratio between the self-interested outcomes (that allow any group of users to deviate jointly) and a social optimum \cite{AFM09,AGTbook, chau03ncg, chau03selfish}. 

\item {\bf Fair Cost-Sharing Mechanisms}:
Sharing economy can be regarded as an alternative to the for-profit sector, which resembles cooperative organizations and favors distributive justice. 
When participants share the costs, there is a notion of fairness. We aim to characterize the inefficiency of distributed decision-making under common fair cost-sharing mechanisms.
First, we consider typical fair cost-sharing concepts such as equal-split, proportional-split, and usage based cost-sharing. Second, we formulate the cost-sharing problem as a bargaining game. Thus, the well-known {\em bargaining game solutions} (e.g., egalitarian and Nash bargaining solutions) \cite{bargainbook} can be applied in the context of sharing economy. 

\end{enumerate}

\begin{table*}[!htb]
\centering
  \caption{A summary of our results. \label{tab:summary}}{
  \begin{tabular}{@{}c|c|c|c|c@{}}
\hline
\hline
    & Equal-split & Proportional-split & Egalitarian/Nash & Usage Based \\
\hline
Existence of Stable Coalition & $\checkmark$ & $\checkmark$ & $\checkmark$ &  Only in some problems\\
&&&&\\
Strong Price of Anarchy & $\Theta(\log K)$ & $\Theta(\log K)$ & $O(\sqrt{K} \log K)$ & $\Omega(K)$ (in general)\\
 & & & & $\Theta(\log K)$ (under certain conditions) \\
\hline
\hline
  \end{tabular}}
\end{table*}

This paper presents a comprehensive study for the SPoA of a general model of coalition formation, considering various common fair cost-sharing mechanisms.

\subsection{Our Contributions}

We consider $K$-capacitated coalitions, where $K$ is the maximum number of sharing participants per coalition. A {\em stable coalition structure}, wherein no group of participants can opt out to form another coalition that leads to lower individual payments, is a likely self-interested outcome. The results in this paper are summarized as follows (and in Table~\ref{tab:summary}). 
\begin{enumerate}

\item The SPoA for any budget balanced cost-sharing mechanism is $O(K)$.

\item However, the SPoA for equal-split and proportional-split cost-sharing is only $\Theta(\log K)$.

\item The SPoA for egalitarian and Nash bargaining solutions is $O(\sqrt{K} \log K)$.

\item The SPoA for usage based cost-sharing is generally $\Omega(K)$. However, we provide natural sufficient conditions to improve the SPoA to be $\Theta(\log K)$, which apply to the examples of sharing economy in this paper.

\item Therefore, distributed decision-making processes under common fair cost-sharing mechanisms induce only moderate inefficiency.

\item We also study the existence of stable coalition structures. 

\end{enumerate}
\section{Related Work} \label{sec:related}

Our problem belongs to the topic of network cost-sharing and hedonic coalition formation problems \cite{albers08ndg,AB12,ADTW08ndg,ADKTWR09ndg,EFM09con,HVW15,H13,RS09,RS14costsharing,KS15sharinggame}. A study particularly related to our results is the price of anarchy for stable matching and the various extensions to $K$-sized coalitions \cite{ADN09,ABH13}. 
Our coalition formation game is a subclass of hedonic coalition formation games \cite{AB12,H13} that allows arbitrary coalitions subject to a constraint on the maximum number of participants per coalition. This model appears to be realistic in many practical settings of sharing economy\footnote{For example, consider taxi-ride sharing. Any passengers can form a coalition to share a taxi-ride, subject to the maximum capacity of a taxicab.}. 

Our work differs from typical cooperative games. In our coalition formation game, each player joins a coalition that incurs a lower individual payment, under a given cost-sharing mechanism. On the other hand, typical cooperative games generally do not consider a specific cost-sharing mechanism, but find a cost-sharing allocation according to certain axioms.


One may regard the results about stable matching in \cite{ADN09,ABH13} as a special case of $K=2$ in our model. However, unlike the stable matching problem, our model has additional structure that can be harnessed for tighter results (e.g., monotonicity). For example, according to \cite{ADN09}, the price of anarchy for Matthew's effect (equivalently, proportional-split cost-sharing) can be unbounded. Here with the help of monotonicity, we show that it is $\Theta(\log K)$. We also study other cost-sharing mechanisms, such as egalitarian, Nash bargaining solutions, and usage based cost-sharing mechanisms that are not considered in \cite{ADN09,ABH13}. Moreover, \cite{ABH13} is based on the comparison of the {\em utility} of a stable matching, while our results are based on the comparison of the {\em cost} of a stable coalition structure. Although it is possible to derive some of our results (for $K=2$) from the previous results, the bounds obtained this way are typically weaker than ours and the gap can increase as a linear function of $K$.  

Network cost-sharing games with capacitated links and non-anonymous cost functions \cite{RS14costsharing, FR12netgames, HF12costsharing} are closely related to our problem. Non-anonymous cost functions may depend on the identity of the players in the coalition, so as to capture the asymmetries between the players because of different service requirements. In \cite{RS14costsharing}, a logarithmic upper bound on the price of anarchy considering the Shapley value in network cost-sharing games with non-anonymous submodular cost functions was given. Our problem can be modeled by a network cost-sharing game with non-anonymous cost functions. In particular, our model is a special case of a $K$-capacitated network cost-sharing game with a simple structure of $n$ parallel links and non-anonymous cost functions such that a strategy profile of the users in the network game corresponds to a coalition structure\footnote{For example, consider taxi-ride sharing. An additional passenger can join a taxi-ride with certain existing passengers who have already formed a coalition, only if all of them will not be worse-off after the change. Otherwise, the existing passengers can always reject the additional passenger by keeping their current coalition. This is always possible in a network cost-sharing game with $n$ parallel links.}. 
However, the key difference of our model from those in \cite{albers08ndg,EFM09con, RS14costsharing} is that we consider replaceable resources from a large pool of  available resources, such that a subset of deviating users can always form a new coalition, independent from other users. This is not true in general network cost-sharing games, when there are limited resources (e.g., links) that a deviating coalition of players can utilize, and it may not be possible to form arbitrary coalitions independent from others. Our model allows us to derive SPoA bounds for diverse cost-sharing mechanisms, whereas only specific cost-sharing mechanisms (e.g., the Shapley value) were considered in general network cost-sharing games.  

It is also worth mentioning how our results of usage based cost-sharing relate to cost-sharing with anonymous cost functions in network design games \cite{albers08ndg} 
or connection games \cite{EFM09con}. On one hand, our model is simpler as we do not assume connectivity requirements in a network, but only an abstract setting that allows arbitrary coalitions up to a certain capacity (but we allow non-anonymous cost functions). On the other hand, one of the cost-sharing mechanisms we consider (i.e., usage based cost-sharing) resembles in some sense that used in \cite{albers08ndg,EFM09con} if we interpret the resources used by one participant as his chosen path or tree in the network design game. Similar to the case in \cite{albers08ndg} (with respect to strong Nash equilibrium), a usage based cost-sharing mechanism may not admit a stable coalition structure. Noteworthily, the SPoA in usage based cost-sharing in our model can increase as a linear function of $K$. Nonetheless, we provide general sufficient conditions for usage based cost-sharing to induce logarithmic SPoA.



\section{Model} \label{sec:model}

This section presents a general model of coalition formation for sharing certain resources, motivated by the applications of sharing economy. Consider an $n$-participant cooperative game in coalition form. The coalitions formed by the subsets of participants are associated with a real-valued cost function. The participants in a coalition are supposed to split the cost according to a certain payment function. However, as a departure from traditional cooperative games, there is a capacity per each coalition, such that not all participants can form a single coalition. Hence, when given a payment function, the participants will opt for coalitions that lead to lower individual payments subject to a capacity per coalition.

\subsection{Problem Formulation}

The set of $n$ participants is denoted by ${\cal N}$. A {\em coalition structure} is a partition of ${\cal N}$ denoted by ${\cal P} \subset 2^{\cal N}$, such that $\bigcup_{G \in \cal P} G = {\cal N}$ and $G_1 \cap G_2 = \varnothing$ for any pair $G_1, G_2 \in {\cal P}$. Let the set of all partitions of ${\cal N}$ be ${\mathscr P}$. Each element $G \in \cal P$ is called a coalition. The set of singleton coalitions, ${\cal P}_{\rm self} \triangleq \{\{i\} : i \in {\cal N} \}$, is called the {\em default coalition structure}, wherein no one forms a coalition with others.  

This paper considers arbitrary coalition structures with at most $K$ participants per coalition, which is motivated by scenarios of sharing replaceable resources; see Section \ref{sec:examples} for examples. The notion of resources will be introduced later in Sec.~\ref{sec:res}, and our model does not always rely on the notion of resources. In practice, $K$ is often much less than $n$. Let ${\mathscr P}_K \triangleq \{{\cal P} \in {\mathscr P} : |G| \le K \mbox{\ for each\ } G \in {\cal P}\}$ be the set of {\it feasible} coalition structures, such that each coalition consists of at most $K$ participants.

\medskip

\subsubsection{Cost Function}
A cost $c(G)$ (also known as a characteristic function) is assigned to each coalition of participants $G \in {\cal P} \in {\mathscr P}_K$, subject to the following properties:
\begin{enumerate}

\item[(C1)] $c(\varnothing) = 0$ and $c(G) > 0$ if $G \ne \varnothing$.

\item[(C2)] {\em Monotonicity}: $c(G) \ge c(H)$, if $H \subseteq G$.

\end{enumerate}
Monotonicity captures natural coalition formation with increasing cost as the number of participants. 
The total cost of coalition structure ${\cal P}$ is denoted by $c({\cal P}) \triangleq \sum_{G \in {\cal P}} c(G)$. For brevity, we also denote $c(\{i\})$ by $c_i$, where $c_i$ is called the {\em default cost} of participant $i$, that is, when $i$ forms no coalition with others. 
When a subset of participants are indexed by ${\cal N'} = \{i_1, i_2, ..., i_j \}\subseteq\cN$, we simply denote the corresponding default costs by $\{ c_1, c_2, ..., c_j\}$. 

A $K$-capacitated {\em social optimum} is a coalition structure ${\cal P}^\ast_K \in {\mathscr P}_K$ that minimizes the total cost:
\begin{align} \hspace{-60pt}
(\textsc{$K$-MinCoalition}) & \ \ {\cal P}^\ast_K \in \arg\min_{{\cal P} \in {\mathscr P}_K} c({\cal P}) \label{eqn:kmincost} 
\end{align}
When $K=2$, a social optimum ${\cal P}^\ast_K$ can be found in polynomial time by reducing the coalition formation problem to a (general graph) matching problem. When $K>2$, \textsc{$K$-MinCoalition} is an NP-hard problem (see Appendix).

\medskip

\subsubsection{Canonical Resources} \label{sec:res}

There is often a resource being shared by each coalition (e.g., a hotel room, a taxicab, a pass). The resources are usually replaceable from a large pool of available resources in the sharing economy. Hence, any subset of participants can always form a coalition using separate resources, independent from other coalitions. For each coalition $G$, we consider a {\em canonical resource}, which is a class of replaceable resources that can satisfy $G$, rather than any specific resource. The canonical resource shared by a coalition will not be affected by the canonical resources shared by other coalitions. Because of the consideration of canonical resources, our model exhibits different properties than the network sharing games with limited resources \cite{RS14costsharing, FR12netgames, HF12costsharing}. 

Let ${\mathscr R}(G)$ be the feasible set of canonical resources that can satisfy coalition $G$. It is naturally assumed that ${\mathscr R}(H) \supseteq {\mathscr R}(G)$, when $H \subseteq G$, because the canonical resources that can satisfy a larger coalition $G$ should also satisfy a smaller coalition $H$ (by ignoring the participants in $G \backslash H$). Each canonical resource $r \in {\mathscr R}(G)$ is characterized by a cost $c_r$, and a set of involved facilities ${\mathscr F}(r)$. Each facility $f \in {\mathscr F}(r)$ carries a cost $c^f$, such that $\sum_{f \in {\mathscr F}(r)} c^f = c_r$. We do not require that every participant of $G$ utilizes the same facilities. Let ${\mathscr F}_i(r)\subseteq {\mathscr F}(r)$ be the set of facilities utilized by participant $i \in G$, when $r$ is shared by $G$. Let $r(G) \in \arg\min_{r \in {\mathscr R}(G)}\{ c_r \}$ be the lowest cost canonical resource for coalition $G$. Hence, we set $c(G) = c_{r(G)}$ and monotonicity is satisfied. If there are multiple lowest cost canonical resources, one is selected by a certain deterministic tie-breaking rule.

\subsection{Motivating Examples}\label{sec:examples}

We present a few motivating examples in sharing economy to illustrate the aforementioned model.

\medskip

\subsubsection{Hotel Room Sharing}

Consider ${\cal N}$ as a set travelers to share hotel rooms. Each traveler $i \in {\cal N}$ is associated with a tuple $(t^{\sf in}_i, t^{\sf out}_i, {\sf A}_i)$, where $t^{\sf in}_i$ is the arrival time, $t^{\sf out}_i$ is the departure time, and ${\sf A}_i$ is the area of preferred locations of hotels. Let $K$ be the maximum number of travelers that can share a room. A canonical resource is a room booking $r$, associated with a tuple $(t^{\sf in}_r, t^{\sf out}_r, a_r)$, where $t^{\sf in}_r$ is the check-in time, $t^{\sf out}_r$ is the check-out time, and $a_r$ is the hotel location. We assume that there is a large pool of available rooms for each location, and we do not consider a specific room. The feasible set ${\mathscr R}(G)$ is a set of room bookings shared by a coalition of travelers $G$, if the following conditions are satisfied:
\begin{enumerate}

\item $|G| \le K$;

\item $a_r \in \bigcap_{i\in G} {\sf A}_i$;

\item $t^{\sf in}_r \le t^{\sf in}_i$ and $t^{\sf out}_r \ge t^{\sf out}_i$ for all $i \in G$.

\end{enumerate}
(Note that the monotonicity assumption is satisfied: ${\mathscr R}(H) \supseteq {\mathscr R}(G)$, when $H \subseteq G$.) 
In this example, ${\mathscr F}(r)$ is be the set of days during $[t^{\sf in}_r, t^{\sf out}_r]$ for room booking $r$, and ${\mathscr F}_i(r)$ be the set of days that $i$ stays in room booking $r$. For each $f \in {\mathscr F}(r)$, $c^f$ is the hotel rate of day $f$.

\medskip

\subsubsection{Taxi-ride Sharing}

Consider ${\cal N}$ as a set of passengers to share taxi-rides. Each passenger $i \in {\cal N}$ is associated with a tuple $(v^{\sf s}_i, v^{\sf d}_i, t^{\sf s}_i, t^{\sf d}_i)$, where $v^{\sf s}_i$ is the source location, $v^{\sf d}_i$ is the destination location, $t^{\sf s}_i$ is the earliest departure time, and $t^{\sf d}_i$ is the latest arrival time. Let $K$ be the maximum number of passengers that can share a ride. A canonical resource is a ride $r$, associated with a sequence of locations $(v^1,...,v^m)$ and a sequence of arrival timeslots $(t^1,...,t^m)$ in an increasing order. We assume that there is a large pool of available taxicabs, and we do not consider a specific taxicab. The feasible set ${\mathscr R}(G)$ is a set of rides shared by a coalition of travelers $G$, if the following conditions are satisfied:
\begin{enumerate}

\item $|G| \le K$;

\item All the locations $(v^{\sf s}_i, v^{\sf d}_i)_{i \in G}$ are present in the sequence $(v^1,...,v^m)$;

\item ${\sf t}_r(v^{\sf s}_i) < {\sf t}_r(v^{\sf d}_i)$, ${\sf t}_r(v^{\sf s}_i) \ge t^{\sf s}_i$ and ${\sf t}_r(v^{\sf d}_i) \le t^{\sf d}_i$ for all $i \in G$, where ${\sf t}_r(v)$ is the arrival timeslot of ride $r$ at location $v$. 

\end{enumerate}
Note that the hotel room sharing problem may be regarded as a one-dimensional version of the taxi-ride sharing problem, if the preferred location constraint is not considered, and we let each tuple $(t^{\sf in}_i, t^{\sf out}_i)$ in hotel room sharing problem be the source and destination locations.
Let ${\mathscr F}(r)$ is the set of road segments traversed by ride $r$, and ${\mathscr F}_i(r)$ be the set of road segments that $i$ travels in ride $r$. For each $f \in {\mathscr F}(r)$, let $c^f$ be the taxi fare for road segment $f$. 

\medskip

\subsubsection{Pass Sharing}

Consider ${\cal N}$ as a set of regular-pass holders who want to form coalitions to share some anonymous passes. Each user $i \in {\cal N}$ is associated with a set of required usage timeslots ${\sf T}_i$. Let $K$ be the maximum number of sharing users, so as to limit the hassle of circulating the pass. A canonical resource is a pass $r$, associated with a set of allowable timeslots ${\sf T}_r$.  We assume that there is a large pool of available passes, and we do not consider a specific pass. The feasible set ${\mathscr R}(G)$ is a set of passes shared by a coalition of travelers $G$, if the following conditions are satisfied:
\begin{enumerate}

\item $|G| \le K$;

\item $T_i\cap T_j=\varnothing \text{ for all }i,j \in G, i\neq j$ (i.e., no one overlaps in their required timeslots);

\item $\bigcup_{i \in G}{\sf T}_i \subseteq {\sf T}_r$.

\end{enumerate}
This setting also applies to sharing physical properties (e.g., houses, cars and parking lots).
Let ${\mathscr F}(r)$ are the set of timeslots required by pass $r$, and $c^f$ be the cost of each timeslot in ${\mathscr F}(r)$. A user needs to cover the cost when he uses the pass or shares the cost with other participants when no one uses it. 
Hence, let ${\mathscr F}_i(r)= {\sf T}_i \cup \big({\sf T}_r\backslash(\bigcup_{j \in G} {\sf T}_j)\big)$, when $i$ shares pass $r$ in coalition $G$.

\medskip

\subsection{Cost-Sharing Mechanisms}

A coalition of participants $G$ are supposed to share the cost $c(G)$. Let the cost (or payment) contributed by participant $i \in G$ be $p_{i}(G)$. 
The utility of participant $i$ is given by: 
\begin{equation} \label{eqn:util}
u_i\big(p_{i}(G)\big) \triangleq c_i - p_{i}(G)
\end{equation}
The following natural properties can be satisfied by payment function $p_{i}(\cdot)$:
\begin{itemize}

\item  {\em Budget Balance}: $p_{i}(\cdot)$ is said to be {\em budget balanced}, if $\sum_{i \in G} p_{i}(G) = c(G)$ for every $G \subseteq {\cal N}$. 

\item  {\em Non-negative Payment}: $p_{i}(\cdot)$ is said to be {\em non-negative}, if $p_i(G) \ge 0$ for every $G \in {\cal P} \in {\mathscr P}_K$. If non-positive payment is allowed, then it possible that $p_i(G) < 0$ for some $i \in G$.

\end{itemize}

This paper considers the following fair cost-sharing mechanisms. Note that only usage based cost-sharing mechanism takes into account the notion of resources, while the other cost-sharing mechanisms do not rely on the notion of resources.
\begin{enumerate}

\item {\bf Equal-split Cost-Sharing}: The cost is split equally among all participants: for $i \in G$,
\begin{equation}
p^{\rm eq}_{i}(G) \triangleq \frac{c(G)}{|G|}
\end{equation}

\item {\bf Proportional-split Cost-Sharing}: The cost is split proportionally according to the participants' default costs: for $i \in G$,
\begin{equation}
p^{\rm pp}_{i}(G)\triangleq \frac{c_i \cdot c(G)}{\sum_{j \in G} c_j} 	
\end{equation}
Namely, $u_i\big(p^{\rm pp}_{i}(G)\big) = c_i \cdot \frac{(\sum_{j \in G} c_j) - c(G)}{\sum_{j \in G} c_j}$. This approach is also called Matthew's effect in \cite{ADN09}. 

\item {\bf Bargaining Based Cost-Sharing}: One can formulate the cost-sharing problem as a bargaining game with a feasible set and a disagreement point. In our model, the feasible set is the set of utilities $(u_i)_{i \in G}$, such that $\sum_{i \in G} u_i \le \sum_{i \in G}c_i - c(G)\ (\Leftrightarrow \sum_{i \in G} p_i \ge c(G))$. The disagreement point is $(u_i = 0)_{i \in G}$, such that each participant pays only the respective default cost. There are two bargaining solutions in the literature \cite{bargainbook}:

\begin{itemize}

\item {\bf Egalitarian Bargaining Solution} is given by:
\begin{equation}
p^{\rm ega}_{i}(G) \triangleq c_i - \frac{(\sum_{j \in G} c_j) - c(G)}{|G|}
\end{equation}
Namely, every participant in each coalition has the same utility: $u_i\big(p^{\rm ega}_{i}(G)\big) = \frac{(\sum_{j \in G} c_j) - c(G)}{|G|}$ for all $i \in G$.
Note that non-positive payment is possible because it may need to compensate those with low default costs to reach equal utility at every participant\footnote{For example, consider $G = \{i, j, k\}$, such that $c_j = c_k = c(G) = 1$ and $c_i = 0.1$. Then $p^{\rm ega}_{i}(G) = -0.26$.}.

\item {\bf Nash Bargaining Solution} is given by:
\begin{equation}
\big( p^{\rm nash}_{i}(G) \big)_{i \in G} \in \arg\max_{(p_{i}(G))_{i \in G}} \prod_{i \in G} u_i\big(p_{i}(G)\big)
\end{equation}
subject to
\begin{equation}
\qquad  \sum_{i \in G} p_{i}(G) = c(G) \notag
\end{equation}
One can impose an additional constraint of non-negative payments: $p_{i}(G) \ge 0$ for all $i \in G$. 

\end{itemize}


\item {\bf Usage Based Cost-Sharing}: Also known as Shapley cost-sharing \cite{albers08ndg}. We consider a cost-sharing mechanism that takes into account the usage structure of resources of participants. Recall that $r(G)$ denotes a lowest-cost canonical resource for coalition $G$. Let ${\cal N}^f(r(G))$ be the set of participants that share the same facility $f$ in $r(G)$. The cost is split equally among the participants who utilize the same facilities:
\begin{equation}
p^{\rm ub}_{i}(G)\triangleq \sum_{f \in {\mathscr F}_i(r(G))} \frac{c^f}{|{\cal N}^f(r(G))|}
\end{equation}
For example, in taxi-ride sharing, passengers will split the cost equally for each road segment with those passengers traveled together in the respective road segment. 

\end{enumerate}

\subsection{Stable Coalition and Strong Price of Anarchy}

Given payment function $p_{i}(\cdot)$, a coalition of participants $G$ is called a {\em blocking coalition} with respect to coalition structure ${\cal P}$ if 
all participants in $G$ can {\it strictly} reduce their payments when they form a coalition $G$ to share the cost instead. A coalition structure is called {\em stable coalition structure}, denoted by $\hat{\cal P}_K \in {\mathscr P}_K$, if there exists no blocking coalition with respect to $\hat{\cal P}_K$. The existence of a stable coalition structure depends on the cost-sharing mechanism (see Appendix). 

Note that a stable coalition structure is also a strong Nash equilibrium\footnote{A strong Nash equilibrium is a Nash equilibrium, in which no group of players can cooperatively deviate in an allowable way that benefits all of its members.} in our model. However, there is a difference between the case of sharing canonical resources and that of limited resources. For sharing canonical resources, an additional member can join an existing coalition to create a larger coalition, only if all of the participants in the new coalition will not be worse-off after the change. Otherwise, the existing coalition can always reject the additional member by keeping the current canonical resource. However, for sharing limited resources, a coalition may be forced to accept an additional member, even they will be worse-off, because they cannot find a new resource to share with. In this case, a strong Nash equilibrium may not be a stable coalition structure.  

Define the {\em Strong Price of Anarchy} (SPoA) as the worst-case ratio between the cost of a stable coalition structure and that of a social optimum over any feasible costs subject to (C1)-(C2):
\begin{equation}
{\sf SPoA}_K \triangleq \max_{c(\cdot),~\hat{\cal P}_K}\frac{c(\hat{\cal P}_K)}{c({\cal P}^\ast_K)}
\end{equation}
Specifically, the strong price of anarchy when using specific cost-sharing mechanisms are denoted by ${\sf SPoA}^{\rm eq}_K$,  ${\sf SPoA}^{\rm pp}_K$, ${\sf SPoA}^{\rm ega}_K$, ${\sf SPoA}^{\rm nash}_K$, ${\sf SPoA}^{\rm ub}_K$, respectively.

\section{Preliminary Results}

Before we derive the SPoA for various cost-sharing mechanisms, we present some preliminary results that will be useful in our proofs. In the following we denote the $K$-th harmonic number by $\cH_K\triangleq\sum_{s=1}^K\frac{1}{s}$.

\begin{customthm}{1}\label{thm:gen}
Recall the default coalition structure ${\cal P}_{\rm self} \triangleq \big\{\{i\} : i \in {\cal N} \big\}$. We have
\begin{equation}
K \cdot c({\cal P}^\ast_K) \ge c({\cal P}_{\rm self})
\end{equation}
Consider a budget balanced payment function $p_{i}(\cdot)$. Let $\hat{\cal P}_K$ be a respective $K$-capacitated stable coalition structure. Then,
\begin{equation}
c({\cal P}_{\rm self}) \ge c(\hat{\cal P}_K)
\end{equation}
Hence, the SPoA for $p_{i}(\cdot)$ is  upper bounded by ${\sf SPoA}_K \le K$.
\end{customthm}
\begin{proof}
First, by monotonicity, we obtain for any $G\in {\cal P}^\ast_K$
\begin{equation}
c(G) \ge \max_{i \in G}\{c_i\} \ge \frac{\sum_{i \in G} c_i}{|G|} \ge \frac{\sum_{i \in G} c_i}{K}
\end{equation}
Hence,
\begin{equation}
 c({\cal P}^\ast_K) = \sum_{G \in {\cal P}^\ast_K} c(G) \ge \frac{1}{K} \sum_{i \in {\cal N}} c_i = \frac{c({\cal P}_{\rm self})}{K}
\end{equation}

Since ${\cal P}^\ast_K$ is a stable coalition structure, then $p_{i}(G) \le c_i$ for every $G \in {\cal P}^\ast_K$. Otherwise, every $i \in G$ can strictly reduce his payment by forming a singleton coalition individually.

Lastly, since $p_{i}(\cdot)$  is a budget balanced payment function, it follows that
\begin{equation}
 c({\cal P}_{\rm self}) = \sum_{i \in {\cal N}}c_i \ge \sum_{G\in\hat{\cal P}_K}\sum_{i\in G} p_{i}(G) = c(\hat{\cal P}_K)
\end{equation}
\end{proof}

However, we will show that the SPoA for various cost-sharing mechanisms is $O(\log K)$ or $O(\sqrt{K}\log K)$.

\medskip
  
To derive an upper bound for the SPoA, the following lemma provides a general tool. First, define the following notation for a non-negative payment function $p_i(\cdot)$:
\begin{equation}\label{sum}
\alpha\big(\{p_{i}(\cdot)\}_{i \in {\cal N}}\big)\triangleq\max_{c(\cdot),~H_1\supset\cdots\supset H_{K}} \frac{\sum_{s=1}^{K} p_{i_s}(H_s)}{c(H_1)},
\end{equation}
where $H_1, ..., H_K$ are a collection of subsets, such that each $H_s \triangleq \{i_s,...,i_{K}\}$ for some $i_1,...,i_K \in {\cal N}$. Note that  $\alpha(\cdot)$ is non-decreasing in $K$. See Appendix for a proof.

\begin{customlemma}{1}\label{lem:common}
Suppose $p_{i}(\cdot)$ is a budget balanced non-negative payment function. Given a $K$-capacitated stable coalition structure $\hat{\cal P}_K$, and a feasible coalition structure ${\cal P} \in {\mathscr P}_K$, then  
$$
\frac{c(\hat{\cal P}_K)}{c({\cal P})} \le \alpha\big(\{p_{i}(\cdot)\}_{i \in {\cal N}}\big)
$$
Thus, if $\hat{\cal P}_K$ is a worst-case stable coalition structure and ${\cal P} = {\cal P}^\ast_K$ is a social optimal coalition structure, then we obtain an upper bound for the SPoA with respect to $\{p_{i}(\cdot)\}_{i \in {\cal N}}$:
$$
{\sf SPoA}_K \le \alpha\big(\{p_{i}(\cdot)\}_{i \in {\cal N}}\big)
$$
\end{customlemma}
\begin{proof}
Let ${\cal P} = \{G_1,...,G_h\}$. Define $H_1^1\triangleq G_1$. Then there exists a participant $i_1^1\in H_1^1$ and a coalition $\hat G_1^1\in\hat{\cal P}_K$, such that $i_1^1\in\hat G_1^1$ and $p_{i_1^1}(H_1^1)\ge p_{i_1^1}(\hat G_1^1)$; otherwise, all the participants in $H_1^1$ would form a coalition $H_1^1$ to strictly reduce their payments, which contradicts the fact that $\hat{\cal P}_K$ is a stable coalition structure.

Next, define $H_2^1\triangleq H_1^1\backslash\{i_1^1\}$. Note that $H_2^1$ is a feasible coalition, because arbitrary coalition structures with at most $K$ participants per coalition are allowed in our model.
By the same argument, there exists $i_2^1\in H_2^1$ and a coalition $\hat G_2^1\in\hat{\cal P}_K$, such that $i_2^1\in\hat G_2^1$ and $p_{i_2^1}(H_2^1)\ge p_{i_2^1}(\hat G_2^1)$.

Let $G_t = \{i^t_1,...,i^t_{K_t}\}$, for any $t\in\{1,...,h\}$. Continuing this argument, we obtain a collection of sets $\{H_s^t\}$, where each $H_s^t \triangleq\{i^t_s,...,i^t_{K_t}\}$ satisfies the following condition:
\begin{quote}

for any $t\in\{1,...,h\}$ and $s\in\{1,..., K_t\}$, there exists $\hat G^t_s\in\hat{\cal P}_K$, such that $i_s^t\in\hat G^t_s$ and 
$p_{i^t_s}(H_s^t)\ge p_{i^t_s}(\hat G^t_s)$

\end{quote}
  
Hence, the SPoA, ${\sf SPoA}_K$, with respect to $\{p_{i}(\cdot)\}_{i \in {\cal N}}$ is upper bounded by
\begin{align}
\frac{c(\hat{\cal P}_K)}{c({\cal P})} &=\frac{\sum_{t=1}^{h} \sum_{s=1}^{K_t}p_{i^t_s}(\hat{G}^t_s)}{\sum_{t=1}^h c(G_t)}\\
&\le \frac{\sum_{t=1}^{h} \sum_{s=1}^{K_t}p_{i^t_s}(H^t_s)}{\sum_{t=1}^hc(H^t_1)}\\
&\le\max_{t\in\{1,...,h\}}\frac{ \sum_{s=1}^{K_t}p_{i^t_s}(H^t_s)}{c(H^t_1)}\le \alpha\big(\{p_{i}(\cdot)\}_{i \in {\cal N}}\big)
\end{align}
because $\alpha(\cdot)$ is non-decreasing in $K$.
\end{proof}

Note that \cite{RS09} uses an approach called {\it summability} similar to that of Lemma~\ref{lem:common}. Informally, a payment function (or cost-sharing mechanism) $p_i(\cdot)$ is said to be {\it $\alpha$-summable} if for every subset $H$ of participants and every possible ordering $\sigma$ on $H$, the sum of the payments of the participants as they are added one-by-one according to $\sigma$ is at most $\alpha\cdot c(H)$.
However, \cite{RS09} relies on the notion of cross-monotonicity for proving summability. A payment function $p_i(\cdot)$ is said to satisfy {\it cross-monotonicity}, if for any $G\subseteq G'$, $p_i(G')\le p_i(G)$. \cite{RS14costsharing} showed that if a payment function satisfies cross-monotonicity in a network cost-sharing game, then summability can bound the price of anarchy. Also, cross-monotonicity implies that a Nash equilibrium is a strong Nash equilibrium. Nonetheless, our model is simpler than network cost-sharing games; Lemma~\ref{lem:common} shows that $\alpha\big(\{p_i(\cdot)\}\big)$ can be used to bound ${\sf SPoA}_K$ without the assumption of cross-monotonicity. In particular, many payment functions may violate cross-monotonicity (e.g., egalitarian, Nash bargaining solution, equal-split and proportional-split), and hence, the approach in \cite{RS09} will not apply to these payment functions. 

\section{Equal-split Cost-Sharing Mechanism}

\begin{customthm}{2} \label{thm:poa_eq}
For equal-split cost-sharing, the SPoA is upper bounded by
$$
{\sf SPoA}^{\rm eq}_K \le \cH_K = \Theta(\log K)
$$
\end{customthm}

\begin{proof} 
Applying Lemma~\ref{lem:common} with $p_i=p_i^{\rm eq}$, we obtain 
\begin{align}
{\sf SPoA}^{\rm eq}_K &\le\alpha\big(\{p_i^{\rm eq}(\cdot)\}_{i \in {\cal N}}\big)\\
& = \max_{c(\cdot),~H_1\supset\cdots\supset H_{K}} \frac{1}{c(H_1)}\sum_{s=1}^{K}\frac{c(H_s)}{K-s+1}\\
& \le \sum_{s = 1}^{K} \frac{1}{s} = \cH_K
\end{align}
which follows from the monotonicity of cost function, $c(H_s)\le c(H_1)$. 
\end{proof}

\subsection{Tight Example} \label{sec:tightexample}

We also present a tight example to show that ${\sf SPoA}^{\rm eq}_K = \Theta(\log K)$. There are $K \cdot K!$ participants, indexed by $${\cal N} = \{i_s^t \mid t = 1,...,K!, s = 1,...,K\}$$ 

For any non-empty subset $G \subseteq {\cal N}$, we define the cost $c(G)$ as follows:
\begin{itemize}

\item Case 1: If $G \subseteq \{ i^t_1,...,i^t_K\}$ for some $t \in \{1,...,K\}$, then we set $c(G) = 1$.

\item Case 2: If $G \subseteq \{ i^{(k-1) \cdot (K-s+1) + 1}_s,...,i^{k \cdot (K-s+1)}_s\}$ for some $k \in \{1 ,..., \frac{K!}{K-s+1}\}$ and $s \in \{1,...,K\}$, then we set $c(G) = 1$. 

\item Case 3: Otherwise, $c(G) = |G|$.


\end{itemize}
It is evident that the preceding setting of cost function $c(\cdot)$ satisfies monotonicity, because  $c(G) = 1$ for cases 1 and 2, otherwise $c(G) = |G| \ge |H| \ge c(H)$ for all $H \subseteq G$.
When $K= 3$, the tight example is illustrated in Fig.~\ref{fig:tight-ex}. 
 
\begin{figure}[!htb]
\center
\includegraphics[width=0.3	\textwidth]{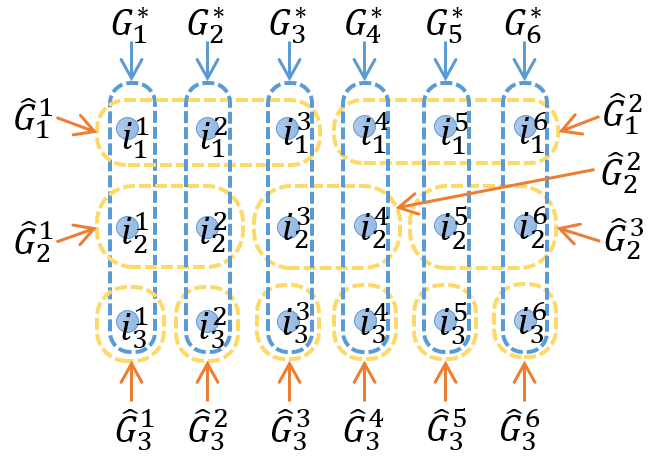}
\caption{An illustration of tight example when $K=3$, where $c(\{i^t_1\}) = c(\{i^t_2\}) = c(\{i^t_3\}) = 1$, $c(\{i^t_1, i^t_2, i^t_3\}) = c(\{i^t_1, i^t_2\}) = c(\{i^t_2, i^t_3\}) = c(\{i^t_1, i^t_3\}) = 1$ for $t=1,...,6$. Also, $c(\{i^1_2, i^2_2\}) = c(\{i^3_2, i^4_2\}) = c(\{i^5_2, i^6_2\}) = 1$ and $c(\{i^1_1, i^2_1, i^3_1\}) = c(\{i^4_1, i^5_1, i^6_1\}) = 1$. The coalitions in orange dotted lines $\{ \hat{G}_s^k \}$ are a stable coalition structure, whereas the coalitions in blue dashed lines $\{ G^\ast_t \}$ are a social optimum. }
\label{fig:tight-ex}
\end{figure} 
 
Let $\hat{G}_s^k \triangleq \{ i^{(k-1) \cdot (K-s+1) + 1}_s,...,i^{k \cdot (K-s+1)}_s\}$, where $s\in\{1,...,K\}, k\in\{1,...,\frac{K!}{K-s+1}\}$. And let $G^\ast_t \triangleq \{i^t_1,...,i^t_K \}$, where $t = \{1,...,K!\}$. See an illustration in Fig.~\ref{fig:tight-ex}.

By equal-split cost-sharing, if $i_s^t \in G$, then
\begin{equation}
p^{\rm eq}_{i_s^t}(G) = 
\left\{
\begin{array}{rl}
\frac{1}{|G|}, & \mbox{\ if\ } G \subseteq \hat{G}_s^k \\
\frac{1}{|G|}, & \mbox{\ if\ } G \subseteq G^\ast_t \\
> \frac{1}{|G|}, &  \mbox{\ otherwise }
\end{array}
\right.
\end{equation}

Note that $p^{\rm eq}_{i_s^t}(\hat{G}_s^k) = p^{\rm eq}_{i_s^t}(G^\ast_t\backslash \{i_1^t,...,i_{s-1}^t\}) = \frac{1}{K-s+1}$. Hence, $i_1^t$ will not switch from coalition $\hat{G}_1^k$ to $G^\ast_t$. Neither will $i_s^t$ switch from coalition $\hat{G}_s^k$ to $G^\ast_t \backslash \{i_1^t,...,i_{s-1}^t\}$.

One can check that $\{ \hat{G}_s^k \}$ are a stable coalition structure, whereas $\{ G^\ast_t \}$ are a social optimum.
Hence, the SPoA is lower bounded by
\begin{equation}
{\sf SPoA}^{\rm eq}_K \ge \frac{\sum_{t=1}^{K!} \sum_{s=1}^{K} p^{\rm eq}_{i_s^t}(\hat{G}_s^k)}{\sum_{t=1}^{K!} c(G^\ast_t)} = \sum_{s = 1}^{K} \frac{1}{K-s+1} = \cH_K 
\end{equation}

\subsection{Proportional-split Cost-Sharing Mechanism}\label{sec:pp}

Given cost function $c(\cdot)$, we define a truncated cost function $\tilde{c}(\cdot)$ as follows: 
\begin{equation} \label{eqn:truncated}
\tilde{c}(G) \triangleq \left\{
\begin{array}{ll}
c(G), & \mbox{\ if\ } c(G) \le \sum_{j \in G} c_j\\
\sum_{j \in G} c_j, & \mbox{\ if\ } c(G) > \sum_{j \in G} c_j\\
\end{array}
\right.
\end{equation}
Note that $\tilde{c}(G) \le \sum_{j \in G} c_j$ for any $G$.

Let ${\sf SPoA}^{\rm pp}_K(c(\cdot))$ be the SPoA with respect to cost function $c(\cdot)$ specifically.

\begin{customlemma}{2} \label{lem:subadd_pp}
For proportional-split cost-sharing, 
$$
{\sf SPoA}^{\rm pp}_K(c(\cdot)) = {\sf SPoA}^{\rm pp}_K(\tilde{c}(\cdot)).
$$
\end{customlemma}

\begin{proof}
First, we show that if $\hat{\cal P}$ is a stable coalition structure, then for any $G \in \hat{\cal P}$, we have $c(G) \le \sum_{j \in G} c_j$. If we assume $c(G) > \sum_{j \in G} c_j$ for some $G \in \hat{\cal P}$, then for all $i \in G$,
\begin{equation}
p^{\rm pp}_{i}(G) = \frac{c_i \cdot c(G)}{\sum_{j \in G} c_j} > c_i
\end{equation} 
Namely, every $i \in G$ can strictly reduce his payment by forming a singleton coalition individually.
This is a contradiction to the fact that $\hat{\cal P}$ is a stable coalition structure.

Second, we note that if ${\cal P}^\ast$ is a social optimum, then for any $G \in {\cal P}^\ast$, we have $c(G) \le \sum_{j \in G} c_j$. Otherwise, ${\cal P}^\ast$ does not attain the least total cost by including $G$.

Therefore, we obtain $\frac{c(\hat{\cal P})}{c({\cal P}^\ast)} = \frac{\tilde{c}(\hat{\cal P})}{\tilde{c}({\cal P}^\ast)}$ and ${\sf SPoA}^{\rm pp}_K(c(\cdot)) = {\sf SPoA}^{\rm pp}_K(\tilde{c}(\cdot))$.
\end{proof}

\begin{customthm}{3} \label{thm:poa_pp1}
For proportional-split cost-sharing, the SPoA is upper bounded by
$$
{\sf SPoA}^{\rm pp}_K \le \log K + 2.
$$
\end{customthm}

\begin{proof}
By Lemma~\ref{lem:subadd_pp} we may assume without loss of generality that  $c(G) \le \sum_{j \in G} c_j$ for any $G$.
Applying Lemma~\ref{lem:common} with $p_i=p_i^{\rm pp}$, we obtain 
\begin{align}
{\sf SPoA}^{\rm pp}_K & \le\alpha\big(\{p_i^{\rm pp}(\cdot)\}_{i \in {\cal N}}\big)\\
& =\max_{c(\cdot),~H_1\supset\cdots\supset H_{K}} \frac{1}{c(H_1)}\sum_{s=1}^{K}p_{i_s}^{\rm pp}(H_s) 
\end{align}
where $H_s = \{i_s, ..., i_K \}$ for some $i_1, ..., i_K \in {\cal N}$, with default costs denoted by $\{c_s, ..., c_K \}$.

Since $p^{\rm pp}_{i_s}(H_s) = \frac{c_s c(H_s)}{\sum_{t=s}^{K} c_t}$, we obtain
\begin{equation}
{\sf SPoA}^{\rm pp}_K \le \max_{\{c_s, c(H_s)\}_{s=1}^{K}} \frac{1}{ c(H_1)} \sum_{s=1}^{K} \frac{c_s c(H_s)}{\sum_{t=s}^{K} c_t} 
\end{equation}
subject to $c(H_1) \ge ... \ge c(H_K)$ and $c(H_s) \ge \max\{c_s,...,c_K\}$ (by monotonicity), and $c(H_s) \le \sum_{t =s}^{K} c_t$ (by Lemma~\ref{lem:subadd_pp}).
Without loss of generality, we assume $c(H_1) = 1$. Hence, $c_s \le c(H_s) \le c(H_1) = 1$. 

Let $\hat{s}$ be the smallest integer, such that $\sum_{t=\hat{s}+1}^{K} c_t \le 1$ (and
hence $\sum_{t=s}^{K} c_t > 1$ for any $s \le \hat{s}$). 
If $s \ge \hat{s}$, we obtain $\frac{c_s c(H_s)}{\sum_{t=s}^{K} c_t} \le c_s$ by our assumption.
If $s < \hat{s}$, we obtain $\frac{c_s c(H_s)}{\sum_{t=s}^{K} c_t} \le \frac{c_s}{\sum_{t=s}^{K} c_t}$.
Therefore,
\begin{align}
{\sf SPoA}^{\rm pp}_K  \le & \max_{\{c_s\}_{s=1}^{K}}\big( \sum_{s=\hat{s}}^{K} c_s + \sum_{s=1}^{\hat{s}-1} \frac{c_s}{\sum_{t=s}^{K} c_t} \big) \\
& \le  2 +  \sum_{s=1}^{\hat{s}-1} \frac{c_s}{\sum_{t=s}^{K} c_t}  \le  2 + \log K
\end{align}
which follows from Lemma~\ref{lem:f_min2} and $c_{\hat{s}} \le 1$. 
\end{proof}

Note that one can strengthen the SPoA by ${\sf SPoA}^{\rm pp}_K \le \cH_K $ for $K\le 6$. We can apply the same tight example in Sec.~\ref{sec:tightexample} to show that ${\sf SPoA}^{\rm pp}_K = \Theta(\log K)$ because $p_i^{\rm pp}(G) = p_i^{\rm eq}(G)$ in this example.

\medskip

\begin{customlemma}{3} \label{lem:f_min2}
If $0 \le c_s \le 1$ for all $s\in\{1,...,K\}$ and $\sum_{s=\hat{s}}^{K}c_s \ge 1$ for some $\hat{s} \le K$, then
$$
\sum_{s=1}^{\hat{s}-1} \frac{c_s}{\sum_{t=s}^{K} c_t} \le \log K
$$
\end{customlemma}

\begin{proof}
First, since $\frac{1}{x+y} \le \frac{1}{x+\delta}$ for any positive numbers $x, y, \delta$, such that $\delta \le y$, we obtain
\begin{equation}
\frac{y}{x+y} \le \int_{0}^{y} \frac{1}{x+\delta} {\rm d} \delta = \log(x+y) - \log x
\end{equation}

It follows that 
\begin{equation}
\frac{c_s}{\sum_{t=s}^{K} c_t} \le \log(\sum_{t=s}^{K} c_t) - \log (\sum_{t=s+1}^{K} c_t)
\end{equation}
Hence,
\begin{align}
\sum_{s=1}^{\hat{s}-1} \frac{c_s}{\sum_{t=s}^{K} c_t} & \le \log(\sum_{s=1}^{K}c_s) - \log(\sum_{s=\hat{s}}^{K}c_s) \\
& \le \log K - \log(1) = \log K
\end{align} 
because $\sum_{s=1}^{K}c_s \le K$ and $\sum_{s=\hat{s}}^{K}c_s \ge 1$.
\end{proof}

\section{Bargaining Based Cost-Sharing Mechanisms}

First, we show that the SPoA for egalitarian bargaining solution and Nash bargaining solution (irrespective of the constraint of non-negative payments) are equivalent. However, there is a difficulty, when we apply Lemma~\ref{lem:common} to obtain an upper bound for the SPoA -- we can only obtain an upper bound as $O(K)$ by the payment function of egalitarian bargaining solution, whereas the payment function of Nash bargaining solution under the constraint of non-negative payments is not convenient to analyze. But we show a property in Nash bargaining solution, namely that there always exists a coalition structure that satisfies positive payments and its cost is at most $(\sqrt{K}+1)$ from that of a given coalition structure. We then obtain an upper bound as $O(\sqrt{K} \log K)$ for the SPoA using this property.

\subsection{Equivalence between Egalitarian and Nash Bargaining}

\begin{customlemma}{4} \label{lem:eganash}
If the constraint of non-negative payments is not considered in Nash bargaining solution, then egalitarian and Nash bargaining solutions are equivalent: $p^{\rm ega}_{i}(G) = p^{\rm nash}_{i}(G)$. 
\end{customlemma}
\begin{proof}
This follows from the fact that the feasible set of Nash bargaining solutions is bounded by the hyperplane $\sum_{i \in G} u_i = \sum_{i \in G}c_i - c(G)$.  The maximal of $\prod_{i \in G} u_i$ (i.e., Nash bargaining solution) is attained at the point $u_i = u_j$ for any $i, j \in G$.
\end{proof}

\begin{customlemma}{5} \label{lem:nashpositive}
For Nash bargaining solution (irrespective of the constraint of non-negative payments), given a stable coalition structure $\hat{\cal P} \in {\mathscr P}_K$ and $G\in \hat{\cal P}$, then every participant has non-negative payment: $p_i^{\rm nash}(G) \ge 0$ for all $i\in G$. 
\end{customlemma}
Namely, each stable coalition structure with the constraint of non-negative payments coincides with a stable coalition structure without the constraint of non-negative payments.
\begin{proof}
Consider a coalition $G = \{i_1, ..., i_K \} \in \hat{\cal P}$, with default costs denoted by $\{c_1, ..., c_K \}$.
We prove the statement by contradiction.
Assume that $p_{i_s}^{\rm nash}(G) < 0$ for some $i_s \in G$. Then, 
\begin{equation}
c_s < \frac{\sum_{j \in G} c_j - c(G)}{|G|} \  \Rightarrow \ 
c_s < \frac{\sum_{j \in G:j\ne i_s} c_j-c(G)}{|G|-1}
\end{equation}
Thus, we obtain
\begin{align}
&\frac{\sum_{j \in G} c_j-c(G)}{|G|} = \frac{c_s+\sum_{j \in G:j\ne i_s} c_j-c(G)}{|G|} \\
< & \frac{(\frac{1}{|G|-1}+1)\big(\sum_{j \in G:j\ne i_s} c_j-c(G)\big)}{|G|}\\
= & \frac{\sum_{j \in G:j\ne i_s} c_j-c(G)}{|G|-1} \\
\le & \frac{\sum_{j \in G:j\ne i_s} c_j-c(G\backslash\{i_s\})}{|G|-1}
\end{align}
where $c(G\backslash\{i_s\}) \le c(G)$ by monotonicity of the cost function. For any $i_k \ne i_s$,
\begin{align}
p_{i_k}^{\rm nash}(G\backslash\{i_s\})&=c_k-\frac{\sum_{j \in G:j\ne i_s} c_j-c(G\backslash\{i_s\})}{|G|-1}\\
&<c_k-\frac{\sum_{j \in G} c_j-c(G)}{|G|}=p_{i_k}^{\rm nash}(G)
\end{align}
Namely, all users in $G\backslash\{i_s\}$ would reduce strictly their payments by switching to the coalition $G\backslash\{i_s\}$. 
This is a contradiction to the fact that $\hat{\cal P}$ is a stable coalition structure.
\end{proof}

\begin{customcor}{1} \label{cor:nashega}
For egalitarian bargaining solution and Nash bargaining solution (irrespective of the constraint of non-negative payments), their SPoA are equivalent:
$${\sf SPoA}^{\rm ega}_K = {\sf SPoA}^{\rm nash}_K$$
\end{customcor}
\begin{proof}
This follows from Lemma~\ref{lem:eganash} and Lemma~\ref{lem:nashpositive}.
\end{proof}

Given cost function $c(\cdot)$, we define a truncated cost function $\tilde{c}(\cdot)$ by Eqn.~(\ref{eqn:truncated}) in Sec.~\ref{sec:pp}, such that $\tilde{c}(G) \le \sum_{j \in G} c_j$ for any $G$.

\begin{customlemma}{6} \label{lem:subadd_ega}
For egalitarian and Nash bargaining solutions (irrespective of the constraint of non-negative payments), we obtain
\begin{align*}
& {\sf SPoA}^{\rm ega}_K(c(\cdot)) = {\sf SPoA}^{\rm nash}_K(c(\cdot)) \\
= \ & {\sf SPoA}^{\rm ega}_K(\tilde{c}(\cdot)) = {\sf SPoA}^{\rm nash}_K(\tilde{c}(\cdot)).
\end{align*}
\end{customlemma}

\begin{proof}
The proof is similar to that of Lemma~\ref{lem:subadd_pp}. First, we show that if $\hat{\cal P}$ is a stable coalition structure, then for any $G \in \hat{\cal P}$, we obtain $c(G) \le \sum_{t \in G} c_t$ based on a contradiction. If we assume $c(G) > \sum_{t \in G} c_t$ for some $G \in \hat{\cal P}$, then for all $i \in G$,
\begin{equation}
p^{\rm ega}_{i}(G) = c_i - \frac{(\sum_{t \in G} c_t) - c(G)}{|G|} > c_i
\end{equation} 
This is a contradiction, since $\hat{\cal P}$ cannot be a stable coalition structure.
Second, we note that if ${\cal P}^\ast$ is a social optimum, then for any $G \in {\cal P}^\ast$, we obtain $c(G) \le \sum_{t \in G} c_t$. 
Therefore, by Corollary~\ref{cor:nashega}, we obtain ${\sf SPoA}^{\rm ega}_K(c(\cdot)) = {\sf SPoA}^{\rm nash}_K(c(\cdot)) = {\sf SPoA}^{\rm ega}_K(\tilde{c}(\cdot)) = {\sf SPoA}^{\rm nash}_K(\tilde{c}(\cdot))$.
\end{proof}

\subsection{Bounding Strong Price of Anarchy}

\begin{customthm}{4} \label{thm:poa_ega1}
For egalitarian bargaining solution and Nash bargaining solution (irrespective of the constraint of non-negative payments),
the SPoA is upper bounded by
\begin{equation}
{\sf SPoA}^{\rm ega}_K = {\sf SPoA}^{\rm nash}_K = O(\sqrt{K} \log K)
\end{equation} 
\end{customthm}

\begin{proof}
First, by Lemma~\ref{lem:subadd_ega}, it suffices to consider egalitarian bargaining solution with cost function satisfying $c(G) \le \sum_{j \in G} c_j$ for any $G$. 

Next, by Lemma~\ref{lem:nashstb1}, there exists a coalition structure $\tilde{\cal P}$, such that $p^{\rm nash}_{i}(G) > 0$ for all $i \in G \in \tilde{\cal P}$, and $c(\tilde{\cal P}) \le (\sqrt{K}+1) \cdot c({\cal P}^\ast_K)$, where ${\cal P}^\ast_K$ is a social optimum.
For a stable coalition structure $\hat{\cal P}_K$, 
$$
\frac{c(\hat{\cal P}_K)}{c({\cal P}^\ast_K)} \le (\sqrt{K}+1) \cdot \frac{c(\hat{\cal P}_K)}{c(\tilde{\cal P})}
$$
Let $\tilde{\cal P} = \{G_1,...,G_h\}$, We apply the similar argument of Lemma~\ref{lem:common} to obtain that for any $t\in\{1,...,h\}$ and $s\in\{1,..., K_t\}$, there exists $\hat G^t_s\in\hat{\cal P}_K$, such that $i_s^t\in\hat G^t_s$ and 
$p_{i^t_s}(H_s^t)\ge p_{i^t_s}(\hat G^t_s)$. Therefore,
\begin{align}
& \frac{c(\hat{\cal P}_K)}{c(\tilde{\cal P})} =\frac{\sum_{t=1}^{h} \sum_{s=1}^{K_t}p^{\rm ega}_{i^t_s}(\hat{G}^t_s)}{\sum_{t=1}^h c(G_t)}\\
\le & \frac{\sum_{t=1}^{h} \sum_{s=1}^{K_t}p^{\rm ega}_{i^t_s}(H^t_s)}{\sum_{t=1}^hc(H^t_1)} \le \alpha\big(\{p^{\rm ega}_{i}(\cdot)\}\big)
\end{align}
Hence, ${\sf SPoA}^{\rm ega}_K \le  (\sqrt{K}+1) \cdot \alpha\big(\{p^{\rm ega}_{i}(\cdot)\}\big)$.

Let $H_s = \{i_s, ..., i_K \}$, with the default costs denoted by $\{c_s, ..., c_K \}$.
Recall that egalitarian bargaining solution is given by
$$
p^{\rm ega}_{i_s}(H_s) = c_s - \frac{(\sum_{t=s}^{K}c_t) - c(H_s)}{K-s+1} 
$$
subject to $c(H_1) \ge ... \ge c(H_K)$ and $c(H_s) \ge \max\{c_s,...,c_K\}$ (by monotonicity), and $c(H_s) \le \sum_{t =s}^{K} c_t$ (by Lemma~\ref{lem:subadd_ega}).  Finally, it follows that ${\sf SPoA}^{\rm eqa}_K \le O(\sqrt{K} \log K)$ by Lemma~\ref{lem:nashstb2}. 
\end{proof}

We can apply the same tight example in Sec.~\ref{sec:tightexample} to show that ${\sf SPoA}^{\rm ega}_K = {\sf SPoA}^{\rm nash}_K \ge \Theta(\log K)$ because $p_i^{\rm ega}(G) = p_i^{\rm eq}(G)$ in this example.

\medskip

\begin{customlemma}{7} (\cite{AEMNP11nash} Theorem 2) \label{lem:nash2}
Consider a coalition $G = \{i_1, ..., i_K \}$, with corresponding default costs denoted by $\{c_1, ..., c_K \}$, such that $c_1 \ge c_2 \ge ... \ge c_{K}$.
If the constraint of non-negative payments is considered, then Nash bargaining solution can be expressed by 
$$
p^{\rm nash}_{i_s}(G) = 
\left\{
\begin{array}{ll}
c_s - \frac{(\sum_{t=1}^{m} c_t) - c(G)}{m}, & \mbox{\ if\ } s \le m \\
0, & \mbox{\ otherwise}
\end{array}
\right.
$$
where $m$ is the largest integer such that $c_m > \frac{(\sum_{t=1}^{m-1} c_t) - c(G)}{m-1}$.
\end{customlemma}

\begin{customlemma}{8} \label{lem:nashstb1}
Suppose $c(G) \le \sum_{j \in G} c_j$ for any $G$.
Given any coalition structure ${\cal P} \in {\mathscr P}_K$, there exists a coalition structure $\tilde{\cal P} \in {\mathscr P}_K$, such that $p^{\rm nash}_{i}(G) > 0$ for all $i \in G \in \tilde{\cal P}$, and $c(\tilde{\cal P}) \le (\sqrt{K} +1)\cdot c({\cal P})$.
\end{customlemma}

\begin{proof}
See Appendix.
\end{proof}

\begin{customlemma}{10} \label{lem:nashstb2}
Let $b_s \triangleq  c(H_s)$.
Consider the following maximization problem:
\begin{align}
& (\textsc{M1}) \ y^\ast(K) \triangleq  \max_{\{c_s, b_s\}_{s=1}^K}  \sum_{s=1}^{K} \Big( c_s - \frac{(\sum_{t=s}^{K}c_t) - b_s}{K-s+1} \Big) \label{P2}\\
& \text{subject to} \notag\\
& \quad b_s\le \sum_{t=s}^Kc_t,~\text{for all } s=1,...,K-1, \\
& \quad  0\le c_s\le b_s\le b_{s+1}\le 1,~\text{for all } s=1,...,K, \\
& \quad b_1 + K c_s - \sum_{t=1}^{K}c_t \ge 0,~\text{for all } s =1,...,K \label{con:positive-}
\end{align}
The maximum of (\textsc{M1}) is upper bounded by $y^\ast(K) \le 1 + \cH_{K-1}  = O(\log K)$.
\end{customlemma}
In Lemma~\ref{lem:nashstb2}, constraint~(\ref{con:positive-}) captures positive payment for every participant in $H_1$. 
\begin{proof}
See Appendix.
\end{proof}

\section{Usage Based Cost-Sharing Mechanism}

Recall that $r(G)$ is the lowest cost canonical resource for $G$, and ${\mathscr F}_i(r(G))$ is the set of facilities utilized by participant $i \in G$ in canonical resource $r(G)$. First, for each subset $L \subseteq G$, we define 
\begin{equation}
X_G(L) \triangleq \sum_{f\in \big(\bigcap_{i \in L} {\mathscr F}_{i}(r(G)) \big) \big\backslash  \big(\bigcup_{i \in G \backslash L}{\mathscr F}_{i}(r(G))\big)} c^f
\end{equation}
Namely, $X_G(L)$ is the total cost of facilities of canonical resource $r(G)$ that are only used by the coalition $L$ exclusively. See Fig.~\ref{fig:usage-ex0} for an illustration of $X_G(L)$.

\begin{figure}[!htb]
\centering
\includegraphics[width=0.27	\textwidth]{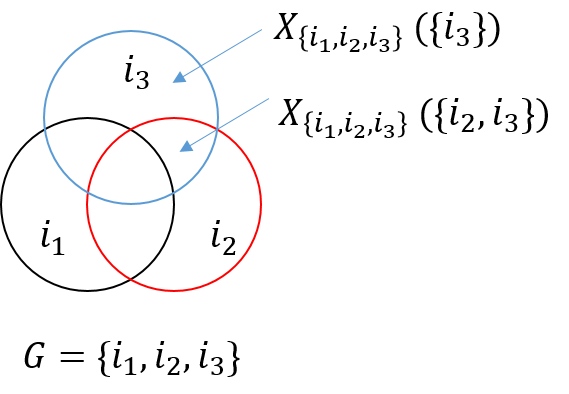}
\caption{Consider $G=\{i_1, i_2, i_3\}$. The three circles depict the sets ${\mathscr F}_{i_1}(r(G)), {\mathscr F}_{i_2}(r(G)), {\mathscr F}_{i_3}(r(G))$. If $c^f$ is represented by a unit area, then $X_G(L)$ is the intersection area of $L$.}
\label{fig:usage-ex0}
\end{figure}

Usage based payment function can be reformulated as
\begin{equation}
p^{\rm ub}_{i}(G) = \sum_{L \subseteq G: i\in L} \frac{X_G(L)}{|L|}
\end{equation}
Given a set of participants $\{i_1, ... i_K \}$ and $H_s \triangleq \{i_s, ..., i_K\}$, we simply write $X_s(L) \triangleq X_{H_s}(L)$.

In general, the strong price of anarchy of usage based cost-sharing can be $\Omega(K)$.

\begin{customthm}{5} \label{thm:poa_ub0}
For usage based cost-sharing in general settings, there exists an instance, such that
$$
{\sf SPoA}^{\rm ub}_K  = \Omega(K).
$$
\end{customthm}

\begin{proof}
See Appendix.
\end{proof}

\subsection{Monotone Utilization}

Generally ${\sf SPoA}^{\rm ub}_K  = \Omega(K)$, but we next present a general sufficient condition for inducing ${\sf SPoA}^{\rm ub}_K = \Theta(\log K)$.
The set of facilities utilized by participants ${\mathscr F}_i(\cdot)$ are said to satisfy {\em monotone utilization}, if for all $H \subseteq G$,
\begin{equation}
\sum_{f \in \cup_{i \in H}{\mathscr F}_{i}(r(H))} c^f \le \sum_{f \in \cup_{i \in H}{\mathscr F}_{i}(r(G))} c^f
\end{equation}
Namely, the total cost of facilities utilized by a subset of participants increases in a larger coalition. Note that monotone utilization condition implies monotonicity of cost function.

For hotel room sharing problem, a set of days of a participant stays in a room booking does not depend on the coalition, and hence, ${\mathscr F}_{i}(r(H)) = {\mathscr F}_{i}(r(G))$. For taxi-ride sharing problem, one needs to travel a greater distance in order to pick-up and drop-off other passengers when sharing with more passengers, and hence, monotone utilization condition is satisfied. Although pass sharing problem does not generally satisfy monotone utilization condition, we will later prove that the SPoA for pass sharing problem with uniform average cost is also $\Theta(\log K)$.

\begin{customthm}{6} \label{thm:poa_ub1}
Consider usage based cost-sharing, such that ${\mathscr F}_i(\cdot)$ satisfies the monotone utilization condition. Then 
\begin{equation}
{\sf SPoA}^{\rm ub}_K \le \cH_K = \Theta(\log K).
\end{equation} 
\end{customthm}
\begin{proof}
See Appendix.
\end{proof}

We can apply the same tight example in Sec.~\ref{sec:tightexample} to show that ${\sf SPoA}^{\rm ub}_K = \Theta(\log K)$. We set $X_G(L)=1$ when $L = G$, otherwise $X_G(L)=0$. This can satisfy monotone utilization condition. It follows that $p_i^{\rm ub}(G) = p_i^{\rm eq}(G)$ in this example.

\subsection{Pass Sharing}

Pass sharing problem violates monotone utilization condition. But we can bound the SPoA for pass sharing problem with uniform average cost $c^f = 1$.

\begin{customthm}{7} \label{thm:poa_ub2}
Consider usage based cost-sharing for pass sharing problem with uniform average cost $c^f = 1$. Then 
\begin{equation}
{\sf SPoA}^{\rm ub}_K \le \cH_K+1=\Theta(\log K)
\end{equation} 
\end{customthm}
\begin{proof}
See Appendix.
\end{proof}


\section{Existence of Stable Coalition Structures} \label{sec:exist}

This section completes the study by investigating the existence of stable coalition structures considering different cost-sharing mechanisms.
First, we define a {\em cyclic preference} as sequences $(i_1,..., i_s)$ and $(G_1,..., G_s)$, where $i_k \in G_k\cap G_{k+1}$ for all $k \le s-1$, and $i_s \in G_s\cap G_1$, such that
\begin{eqnarray*}
u_{i_1}(p_{i_1}(G_1)) & > & u_{i_1}(p_{i_1}(G_2)), \\
u_{i_2}(p_{i_2}(G_2)) & > & u_{i_2}(p_{i_2}(G_3)), \\
& \vdots & \notag \\
u_{i_s}(p_{i_s}(G_s)) & > & u_{i_s}(p_{i_s}(G_1))
\end{eqnarray*}

\begin{customlemma}{11} \label{lem:cyc}
If there exists no cyclic preference, there always exists a stable coalition structure. Furthermore, such a stable coalition structure can be found in time $n^{O(K)}$.
\end{customlemma}

\begin{proof}
See Appendix.
\end{proof}

\medskip

By Lemma~\ref{lem:cyc}, we can show that there always exists a stable coalition structure for equal-split, proportional-split cost sharing mechanisms, egalitarian bargaining solution and Nash bargaining solution. See Appendix for full proofs.
In general, usage based cost-sharing can induce cyclic preference, and hence, possibly the absence of a stable coalition structure. However, we show the existence of a stable coalition structure in some special cases, for example, pass sharing problem and hotel room sharing problem when $K=2$ (see Appendix).
\vspace{-5pt}
\section{Conclusion}\label{sec:concl}

Sharing economy is a popular paradigm for social and economic interactions with distributed decision-making processes. Motivated by the applications of sharing economy, this paper studies a coalition formation game with a constraint on the maximum number of sharing participants per each coalition. This coalition formation game can capture a number applications of sharing economy, such as hotel room, taxi-ride and pass sharing problems. A number of cost-sharing mechanisms are considered, wherein each participant is interested in joining a coalition with a lower payment in the respective cost-sharing mechanism. We study stable coalitions, wherein no coalition of participants can deviate unilaterally to form lower cost coalitions, as the likely self-interested outcomes. 

To quantify the inefficiency of distributed decision-making processes, we show that the Strong Price of Anarchy (SPoA) between a worst-case stable coalition and the social optimum for egalitarian and Nash solutions is $O(\sqrt{K} \log K)$, whereas the one for equal-split, proportional-split, and usage based cost-sharing (under monotone consumption condition or for pass sharing problem) is only $\Theta(\log K)$, where $K$ is the maximum capacity of sharing participants. Therefore, distributed decision-making processes under common fair cost-sharing mechanisms induce only moderate inefficiency.

The SPoA for egalitarian and Nash solutions (i.e., $O(\sqrt{K} \log K)$) is not known to be tight. It is interesting to see if the gap will be closed. Furthermore, the SPoA for specific problems (e.g., hotel room, taxi-ride and pass sharing problems) may be strictly smaller than $\Theta(\log K)$. A companion study of empirical SPoA for taxi-ride sharing using real-world taxi data can be found in \cite{CTE16}.

\bibliographystyle{ieeetr}
\bibliography{reference}

\begin{thebibliography}{10}

\bibitem{economist13}
{The Economist}, ``The rise of the sharing economy,'' Mar. 2013.

\bibitem{AFM09}
N.~Andelman, M.~Feldman, and Y.~Mansour, ``Strong price of anarchy,'' {\em
  Games and Economic Behavior}, vol.~65, no.~2, pp.~289 -- 317, 2009.

\bibitem{AGTbook}
N.~Nisan, T.~Roughgarden, E.~Tardos, and V.~V. Vazirani, {\em Algorithmic Game
  Theory}.
\newblock Cambridge University Press, 2007.

\bibitem{chau03ncg}
C.-K. Chau and K.~M. Sim, ``The price of anarchy for non-atomic congestion
  games with symmetric cost maps and elastic demands,'' {\em Operations
  Research Letters}, vol.~31, pp.~327--334, Sept. 2003.

\bibitem{chau03selfish}
C.-K. Chau and K.~M. Sim, ``Analyzing the impact of selfish behaviors of
  internet users and operators”,'' {\em IEEE Communications Letters}, vol.~7,
  pp.~463--465, Sept. 2003.

\bibitem{bargainbook}
A.~E. Roth, {\em Game-Theoretic Models of Bargaining}.
\newblock Cambridge University Press, 2005.

\bibitem{albers08ndg}
S.~Albers, ``On the value of coordination in network design,'' {\em SIAM
  Journal on Computing}, vol.~38, pp.~2273–--2302, 2009.

\bibitem{AB12}
H.~Aziz and F.~Brandl, ``Existence of stability in hedonic coalition formation
  games,'' in {\em Proc. of Intl. Conf. on Autonomous Agents and Multiagent
  Systems (AAMAS)}, 2012.

\bibitem{ADTW08ndg}
E.~Anshelevich, A.~Dasgupta, E.~Tardos, and T.~Wexler, ``Near-optimal network
  design with selfish agents,'' {\em Theory of Computing}, vol.~4, pp.~77--109,
  2008.

\bibitem{ADKTWR09ndg}
E.~Anshelevich, A.~Dasgupta, J.~Kleinberg, E.~Tardos, and T.~Wexler, ``The
  price of stability for network design with fair cost allocation,'' {\em SIAM
  Journal on Computing}, vol.~38, pp.~1602--1623, 2009.

\bibitem{EFM09con}
A.~Epsteina, M.~Feldmanb, and Y.~Mansour, ``Strong equilibrium in cost sharing
  connection games,'' {\em Games and Economic Behavior}, vol.~67,
  pp.~51–--68, 2009.

\bibitem{HVW15}
M.~Hoefer, D.~Vaz, and L.~Wagner, ``Hedonic coalition formation in networks,''
  in {\em {AAAI} Conf. on Artificial Intelligence}, pp.~929--935, 2015.

\bibitem{H13}
M.~Hoefer, ``Strategic cooperation in cost sharing games,'' {\em International
  Journal of Game Theory}, vol.~42, no.~1, pp.~29--53, 2013.

\bibitem{RS09}
T.~Roughgarden and M.~Sundararajan, ``Quantifying inefficiency in cost sharing
  mechanisms,'' {\em J. {ACM}}, vol.~56, no.~4, 2009.

\bibitem{RS14costsharing}
T.~Roughgarden and O.~Schrijvers, ``Network cost sharing without anonymity,''
  in {\em Algorithmic Game Theory, Lecture Notes in Computer Science},
  pp.~134--145, 2014.

\bibitem{KS15sharinggame}
M.~Klimm and D.~Schmand, ``Sharing non-anonymous costs of multiple resources
  optimally,'' in {\em Proc. of Intl. Conf. on Algorithms and Complexity
  (CIAC)}, 2015.

\bibitem{ADN09}
E.~Anshelevich, S.~Das, and Y.~Naamad, ``Anarchy, stability, and utopia:
  Creating better matchings,'' in {\em Algorithmic Game Theory {SAGT}},
  pp.~159--170, 2009.

\bibitem{ABH13}
E.~Anshelevich, O.~Bhardwaj, and M.~Hoefer, ``Friendship and stable matching,''
  in {\em Algorithms - {ESA}}, pp.~49--60, 2013.

\bibitem{FR12netgames}
M.~Feldman and T.~Ron, ``Capacitated network design games,'' in {\em
  Algorithmic Game Theory, Lecture Notes in Computer Science}, pp.~132--143,
  2012.

\bibitem{HF12costsharing}
T.~Harks and P.~von Falkenhausen, ``Optimal cost sharing for capacitated
  facility location games,'' {\em European Journal of Operational Research},
  vol.~239, pp.~187–--198, 2014.

\bibitem{AEMNP11nash}
K.~Avrachenkov, J.~Elias, F.~Martignon, G.~Neglia, and L.~Petrosyan, ``A {Nash}
  bargaining solution for cooperative network formation games,'' in {\em
  Networking, Lecture Notes in Computer Science}, pp.~307--318, 2011.

\bibitem{CTE16}
C.-K. Chau, C.-M. Tseng, and K.~Elbassioni, ``How to share taxi fare: A network
  analysis of decentralized coalition formation in sharing economy,'' tech.
  rep., Masdar Institute, 2016.

\bibitem{KSPbook}
H.~Kellerer, U.~Pferschy, and D.~Pisinger, {\em Knapsack Problems}.
\newblock Springer, 2004.

\end{thebibliography}

\clearpage

\section*{Appendix} 
\subsection{Monotone Property of $\alpha(\cdot)$}
We show that $\alpha(\cdot)$ is non-decreasing in $K$. Recall that 
\begin{equation}
\alpha_K\big(\{p_{i}(\cdot)\}_{i \in {\cal N}}\big)\triangleq\max_{c(\cdot),~H_1\supset\cdots\supset H_{K}} \frac{\sum_{s=1}^{K} p_{i_s}(H_s)}{c(H_1)},
\end{equation}
where $H_1, ..., H_K$ are a collection of subsets, such that each $H_s \triangleq \{i_s,...,i_{K}\}$ for some $i_1,...,i_K \in {\cal N}$. We explicitly use subscript $K$ to denote the dependence of $K$. For brevity, consider $K = K' + 1 $. Suppose
\begin{equation}
\alpha_{K'}=  \frac{\sum_{s=1}^{K'} p_{i_s}(H_s)}{c(H_1)}
\end{equation}
for a cost function $c(\cdot)$ and some subsets $H_1, ..., H_{K'}$, where $|H_1|=K'$. Then we construct a new set $H_1 \cup \{ i \}$ for some element $i$, and let $c(H_1 \cup \{ i \}) = c(H_1)$ (which still satisfies monotonicity). Then
$$
\alpha_{K} \ge \frac{p_{i}(H_1 \cup \{ i \}) + \sum_{s=1}^{K'} p_{i_s}(H_s)}{c(H_1 \cup \{ i \})} \ge \alpha_{K'}
$$
because $p_{i}(\cdot)$ is non-negative.

\subsection{Bargaining Based Cost-Sharing Mechanisms}

\begin{customlemma}{7} (\cite{AEMNP11nash} Theorem 2) \label{lem:nash2}
Consider a coalition $G = \{i_1, ..., i_K \}$, with corresponding default costs denoted by $\{c_1, ..., c_K \}$, such that $c_1 \ge c_2 \ge ... \ge c_{K}$.
If the constraint of non-negative payments is considered, then Nash bargaining solution can be expressed by 
$$
p^{\rm nash}_{i_s}(G) = 
\left\{
\begin{array}{ll}
c_s - \frac{(\sum_{t=1}^{m} c_t) - c(G)}{m}, & \mbox{\ if\ } s \le m \\
0, & \mbox{\ otherwise}
\end{array}
\right.
$$
where $m$ is the largest integer such that $c_m > \frac{(\sum_{t=1}^{m-1} c_t) - c(G)}{m-1}$.
\end{customlemma}

\begin{customlemma}{8} \label{lem:nashstb1}
Suppose $c(G) \le \sum_{j \in G} c_j$ for any $G$.
Given any coalition structure ${\cal P} \in {\mathscr P}_K$, there exists a coalition structure $\tilde{\cal P} \in {\mathscr P}_K$, such that $p^{\rm nash}_{i}(G) > 0$ for all $i \in G \in \tilde{\cal P}$, and $c(\tilde{\cal P}) \le (\sqrt{K} +1)\cdot c({\cal P})$.
\end{customlemma}

\begin{proof}
Given coalition structure ${\cal P}$, we construct coalition structure $\tilde{\cal P}$ as follows. For each $G \in {\cal P}$, we sort the participants $G = \{i_1,...,i_m \}$ in the decreasing order of their default costs, such that $c_1 \ge c_2 \ge ... \ge c_m$.
We next split $G$ into $R$ sub-groups $\{N_1,..., N_R\}$, where $i_1 ,..., i_{m_1} \in N_1$, $i_{m_1 + 1} ,..., i_{m_2} \in N_2, \cdots, i_{m_{R-1} + 1} ,..., i_{m_R} \in N_R$, such that for $k=m_{t-1}+1,..., m_t$ 
\begin{equation}
c_k - \frac{\sum_{l = m_{t-1}+1}^{k} c_l - c(\{i_{m_{t-1}+1},...,i_k\})}{k - m_{t-1}} > 0  \label{e111}
 \end{equation}
\begin{equation}
c_{i_{m_t+1}} - \frac{\sum_{l = m_{t-1}+1}^{m_t} c_l - c(\{i_{m_{t-1}+1},...,i_{m_t+1}\})}{m_t - m_{t-1}} \le 0 \label{e112}
\end{equation}
where $t \in \{1, ..., R \}$ and $m_0 = 0$.
By monotonicity of $c(\cdot)$ and the ordering on $c_i$, Eqns.~\raf{e111}-\raf{e112} imply that each $i_k \in N_t$ satisfies
$$
c_k - \frac{\sum_{l = m_{t-1}}^{m_{t}+1} c_l - c(N_t)}{m_{t} - m_{t-1}} > 0, 
$$
and
$$
c_k - \frac{\sum_{l = m_{t-2}+1}^{m_{t-1}} c_l - c(N_{t-1})}{m_{t-1} - m_{t-2}} \le 0
$$
By Lemma~\ref{lem:nash2}, the above conditions can guarantee positive payments in Nash bargaining solution. We then replace each coalition $G \in {\cal P}$ by a collection of coalitions $N_1,..., N_R$. We call such a coalition structure $\tilde{\cal P}$.

Note that for each $i_k \in N_{t+1}$, 
\begin{equation}
c_k \le \frac{\sum_{l = m_{t-1}+1}^{m_{t}} c_l - c(N_t)}{m_{t} - m_{t-1}} \le \frac{n_{t} c_{m_{t-1}+1} - c(N_t)}{n_{t}} 
\end{equation}
where $n_t \triangleq |N_t|=m_{t} - m_{t-1}$.
Without loss of generality, we assume 
$c(G) = 1$. Let $\bar{c}_t \triangleq c_{m_{t-1} + 1}$ and $C_t \triangleq c(N_t)$. 
It is evident to see that $1 \ge \bar{c}_1 > \bar{c}_2 > ... > \bar{c}_R>0$.
Since $c(G) \le \sum_{j \in G} c_j$ for any $G$, $C_t \le n_t \bar{c}_t$ for all $t \in \{1, ..., K\}$. 

We next upper bound $\sum_{t=1}^{R} C_t$ by the following optimization problem (\textsc{S1}):
\begin{align}
(\textsc{S1}) \quad &\max_{(C_t,n_t)_{t=1}^R}  \sum_{t=1}^{R}  C_t\\
\text{subject to}& \quad 0\le C_t \le 1, ~\text{for all } t \in \{1, ..., R\}\\
 &\quad \sum_{t=1}^{R} n_t \le K \\
 & \quad  C_t\le n_t(\bar{c}_t - \bar{c}_{t+1}), \text{ for all }t \in \{1, ..., R\},\label{e1111}
\end{align}
where we assume $\bar{c}_{R+1}=0$.
Since in (\textsc{S1}) the lower bounds on $n_t$ are only present in Constraints \raf{e1111}, we assume 
$
n_t=\frac{C_t}{y_t},  
$
where ${y_t}\triangleq\bar{c}_t - \bar{c}_{t+1}$ for $t=1,...,R$, and obtain
\begin{align}
(\textsc{S2}) \quad &\max_{(C_t)_{t=1}^R}  \sum_{t=1}^{R}  C_t\\
\text{subject to}& \quad 0\le C_t \le 1, ~\text{for all } t \in \{1, ..., R\}\\
 &\quad \sum_{t=1}^{R} \frac{C_t}{y_t} \le K
\end{align}
Note that (\textsc{S2}) is simply a {\it fractional knapsack problem}. Suppose that $(y_t)_{t=1}^R$ are arranged in a {\it non-increasing} order, $y_{1}\ge y_{2}\ge\cdots\ge y_{R}$. Let $\ell$ be the largest index such that 
\begin{equation}
\sum_{t=1}^{\ell}\frac{1}{y_{t}}\le K~~\text{ and }~~\sum_{t=\ell+1}^{R}\frac{1}{y_{t}}> K
\end{equation}
Then the optimal solution $(C_t^\ast)_{t=1}^R$ to (\textsc{S2}) is given by Lemma~\ref{lem:fkp}.
Hence, the optimal value of (\textsc{S2}) is $\sum_{t=1}^R C_t^\ast \le \ell+1$.  
Note that 
\begin{equation}
\sum_{t=1}^Ry_t=\sum_{t=1}^R(\bar{c}_t-\bar{c}_{t+1})=\bar{c}_1\le 1
\end{equation}
By the arithmetic mean-harmonic mean inequality (i.e., $\sum_{t=1}^\ell \frac{y_{t}}{\ell} \ge \frac{\ell}{\sum_{t=1}^\ell  \frac{1}{y _{t}}}$), we obtain
\begin{equation}
\frac{\ell}{K}\le\frac{\ell}{\sum_{t=1}^\ell  \frac{1}{y _{t}}}\le\sum_{t=1}^\ell \frac{y_{t}}{\ell}\le \frac{1}{\ell}
\end{equation}
Hence, it follows that $\ell\le \sqrt{K}$, and the maximum of (\textsc{S2}) is upper bounded by $\sqrt{K}+1$.

Therefore, this completes the proof by
\begin{equation}
\sum_{G \in {\cal P}} \sum_{t=1}^R c(N_t) \le \sum_{G \in {\cal P}} (\sqrt{K}+1) \cdot c(G) \ \Rightarrow \ c(\tilde{\cal P}) \le (\sqrt{K}+1)  \cdot c({\cal P}).
\end{equation}
\end{proof}

\begin{customlemma}{9} \label{lem:fkp}
The fractional knapsack problem is defined by
\begin{align}
(\textsc{FKP}) \quad &\max_{(C_t)_{t=1}^R}  \sum_{t=1}^{R}  C_t\\
\text{subject to}& \quad 0\le C_t \le 1, ~\text{for all } t \in \{1, ..., R\}\\
 &\quad \sum_{t=1}^{R} \frac{C_t}{y_t} \le K
\end{align}
Suppose that $(y_t)_{t=1}^R$ are positive and arranged in a {\it non-increasing} order, $y_{1}\ge y_{2}\ge\cdots\ge y_{R}$. Let $\ell$ be the largest index such that 
\begin{equation}
\sum_{t=1}^{\ell}\frac{1}{y_{t}}\le K~~\text{ and }~~\sum_{t=\ell+1}^{R}\frac{1}{y_{t}}> K
\end{equation}
Then the optimal solution $(C_t^\ast)_{t=1}^R$ to \textsc{(FKP)} is given by 
$$
C_{t}^\ast=\left\{\begin{array}{ll}
1, &\text{ if\ } t \in \{1,...,\ell\},\\
\min\{K-\sum_{j=1}^\ell\frac{1}{y_{j}},1\}, &\text{ if\ } t=\ell+1,\\
0, &\text{ if\ } t \in \{\ell+2,...,R\}
\end{array}
\right.
$$
\end{customlemma}
\begin{proof}
The proof follows from a well-known result in knapsack problem (for example, see \cite{KSPbook} Theorem 2.2.1).
\end{proof}

\begin{customlemma}{10} \label{lem:nashstb2}
Let $b_s \triangleq  c(H_s)$.
Consider the following maximization problem.
\begin{align}
& (\textsc{M1}) \ y^\ast(K) \triangleq  \max_{\{c_s, b_s\}_{s=1}^K}  \sum_{s=1}^{K} \Big( c_s - \frac{(\sum_{t=s}^{K}c_t) - b_s}{K-s+1} \Big) \label{P2}\\
& \text{subject to} \notag\\
& \quad b_s\le \sum_{t=s}^Kc_t,~\text{for all } s\in\{1,...,K-1\}, \\
 &\quad 0\le c_s\le b_s\le b_{s+1}\le 1,~\text{for all } s\in\{1,...,K\}, \\
 &\quad b_1 + K c_s - \sum_{t=1}^{K}c_t \ge 0,~\text{for all } s\in\{1,...,K\}. \label{con:positive}
\end{align}
The maximum is upper bounded by $y^\ast(K) \le 1 + \sum_{s=1}^{K-1}\frac{1}{s} = O(\log K)$.
\end{customlemma}
\begin{proof}
Clearly, it is enough to bound (\textsc{M1}) subject to the relaxation:
\begin{eqnarray}
0\le c_s\le b_s\le b_{s+1}\le 1,~\text{for all } s\in\{1,...,K\},\label{eee-1}\\
b_1 + K c_s - \sum_{t=1}^{K}c_t \ge 0,~\text{for all } s\in\{1,...,K\}.\label{eee-2}
\end{eqnarray}
As the coefficients of $b_s$ in (\textsc{M1}) are positive, it is clear that setting $b_s=1$, for all $s$, will maximize  (\textsc{M1}) without violating the Constraints (\ref{eee-1}) and (\ref{eee-2}). It is then enough to show that the optimum value $z^\ast$ of the following linear program is at most $1-\frac{1}{K}$.
\begin{align}
& z^\ast \triangleq\max_{c_s} \sum_{s=1}^{K} \frac{(K-s+1) c_s - \sum_{t=s}^{K}c_t}{K-s+1} \label{p1}\\
& \text{subject to} \notag\\
&\quad \max\left\{0,\frac{\sum_{t=1}^Kc_t - 1}{K}\right\} \le c_s \le  1, ~\text{for all } s\in\{1,...,K\}.\label{eeee-1}
\end{align}
We write 
\begin{align}\label{ee1}
f(c_1,...,c_K)&\triangleq \sum_{s=1}^{K} \frac{(K-s+1) c_s - \sum_{t=s}^{K}c_t}{K-s+1} \\
& =\sum_{s=1}^Kc_s\alpha_s,
\end{align}
where $\alpha_s\triangleq1-\sum_{t=1}^{s}\frac{1}{K-t+1}$. 
Denote by
$c_s^\ast$, for $s=1,...,K,$ the optimal values maximizing $f(c_1,...,c_K)$ subject to the constraints given in Eqn.~(\ref{eeee-1}).  
Note that 
\begin{align}\label{eee1}
\sum_{s=1}^{K}\alpha_s&=K-\sum_{s=1}^{K}\sum_{t=1}^s\frac{1}{K-s+1}\\
&=K-\sum_{t=1}^{K}\frac{1}{K-s+1}\sum_{s=t}^K 1=0
\end{align}

\begin{claim}
Let $(c_1,...,c_K)\ne \mathbf{0}$ be a basic feasible solution (BFS) of LP Eqn.~(\ref{p1}). Then, there exists a partition $S_1\cup S_o$ of $[K]$ such that and $c_s=1$ for all $s\in S_1$, and $c_s=\bar c\triangleq\frac{K-h-1}{K-h}$ for all $s\in S_o$, where $h\triangleq|S_o|\le K-1$.
\end{claim}
\begin{proof}
A BFS $(c_1,...,c_K)$ of the LP Eqn.~(\ref{p1}) is determined by exactly $K$ equations among the inequalities in Eqn.~(\ref{eeee-1}).

We first note that we must have $\max\left\{0,\frac{\sum_{t=1}^Kc_t - 1}{K}\right\}=\frac{\sum_{t=1}^Kc_t - 1}{K}$; otherwise, $\sum_{t=1}^Kc_t<1$ implies that $c_s=0$ for all $s$ (since the inequalities in Eqn.~(\ref{eeee-1}) in this case reduce to $0\le c_s\le 1$ and the BFS will pick one of these two inequalities for each $s$).

Let $S_1\triangleq\{s:c_s=1\}$ and $S_o\triangleq\{s:c_s=\frac{\sum_{t=1}^Kc_t - 1}{K}\}$, and note that $S_1\cup S_o=\{1,...,K\}$. Write $x\triangleq\frac{\sum_{t=1}^{K}c_t-1}{K}=\frac{hx+K-h-1}{K}$. Then $x=\bar c$.
\end{proof}
Let $(c_1,...,c_K)$ be a BFS of Eqn.~(\ref{eeee-1}) defined by the partition $S_o\cup S_1$. Then substituting the value of $c_s$ in Eqn.~(\ref{ee1}) and using Eqn.~(\ref{eee1}) we obtain
\begin{eqnarray}\label{eeeee-1}
f(c_1,...,c_K)&=&\bar c\sum_{s\in S_o}\alpha_s+\sum_{s\in S_1}\alpha_s\nonumber\\
&=&(1-\bar c)\sum_{i\in S_1}\alpha_s=\frac{\sum_{s\in S_1}\alpha_s}{|S_1|}
\end{eqnarray}
Note that $\alpha_s>\alpha_{s+1}$ for all $s=1,...,K-1$. 
It follows that the choice $S_1=\{1\}$ maximizes $f(c_1,...,c_K)$ and the lemma follows.
\end{proof}

\subsection{Usage Based Cost-Sharing Mechanism}

\begin{customthm}{5} \label{thm:poa_ub0}
For usage based cost-sharing in general settings, there exists an instance, such that
$$
{\sf SPoA}^{\rm ub}_K  = \Omega(K).
$$
\end{customthm}

\begin{proof}
We construct a similar instance to the one in Sec.~\ref{sec:tightexample}. There are $K \cdot K!$ participants, indexed by ${\cal N} = \{i_s^t \mid t = 1,...,K, s = 1,...,K!\}$. 
For any non-empty subset $G \subseteq {\cal N}$ and $L \subseteq G$, we define $X_G(L)$ as follows: 
\begin{itemize}

\item Case 1: If $G \subseteq \{ i^t_1,...,i^t_K\}$ for some $t \in \{1,...,K\}$, then 
\begin{itemize}

\item
If $|G| = 1$, then $X_G(L) = 1$.

\item 
If $|G| > 1$, then let $s = \min\{s' \mid i^t_{s'} \in G\}$, and
\begin{equation}
X_G(L) = 
\left\{
\begin{array}{rl}
\frac{1}{2}, & \mbox{if\ } L = \{ i^t_s \}\\
\frac{1}{2}, & \mbox{if\ } L = G\backslash\{ i^t_s \}\\
0, & \mbox{otherwise\ } \\
\end{array}
\right.
\end{equation}

\end{itemize}
See an illustration of an example for setting $X_G(L)$ in Fig.~\ref{fig:usage-ex}.

\begin{figure}[!htb]
\centering
\includegraphics[width=0.5	\textwidth]{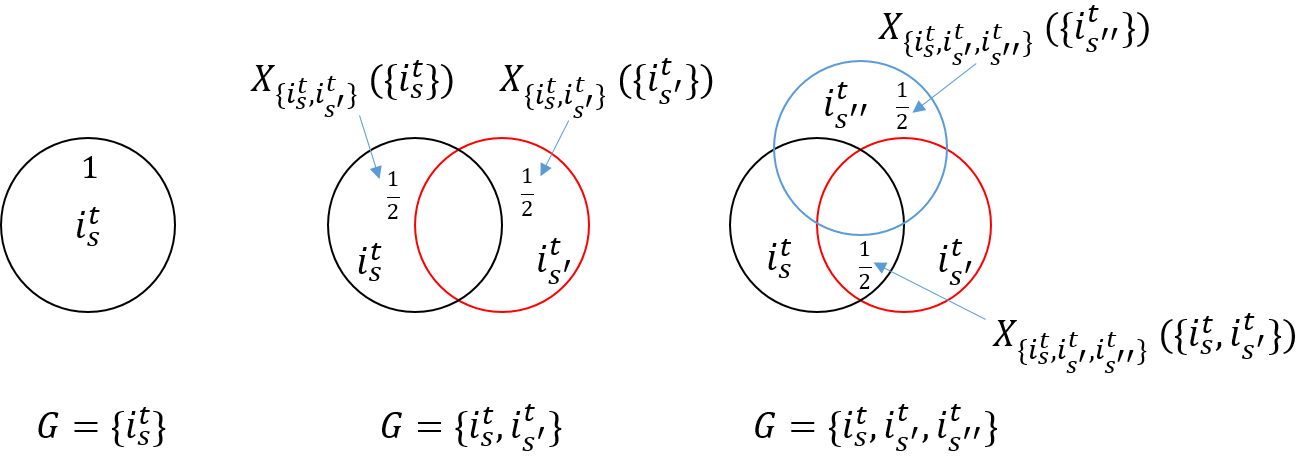}
\caption{An illustration of an example for setting $X_G(L)$ in case 1. Suppose $s'' < s' < s$. $X_{\{ i_{s}^t \}}\big(\{ i_{s}^t \}\big) = 1, X_{\{ i_{s}^t, i_{s'}^t \}}\big(\{ i_{s}^t \}\big) = X_{\{ i_{s}^t, i_{s'}^t \}}\big(\{ i_{s'}^t \}\big) = X_{\{ i_{s}^t, i_{s'}^t, i_{s''}^t \}}\big(\{ i_{s''}^t \}\big) = X_{\{ i_{s}^t, i_{s'}^t, i_{s''}^t \}}\big(\{ i_{s}^t, i_{s'}^t \}\big) = \frac{1}{2}$. Otherwise, $X_G(L) = 0$.}
\label{fig:usage-ex}
\end{figure}

\item Case 2:  If $G \subseteq \{ i^{(k-1) \cdot (K-s+1) + 1}_s,...,i^{k \cdot (K-s+1)}_s\}$ for some $k \in \{1 ,..., \frac{K!}{K-s+1}\}$ and $s \in \{1,...,K\}$, then 
\begin{equation}
X_G(L) = 
\left\{
\begin{array}{rl}
1, & \mbox{if\ } L = G\\
0, & \mbox{otherwise\ } \\
\end{array}
\right.
\end{equation}

\item Case 3: Otherwise, $G$ is a partition of $m$ sets $G = H_1 \cup \cdots \cup H_m$, where each $H_t$ is a maximal subset that satisfies either Case 1 or Case 2. Then, 
\begin{equation}
X_G(L) = \sum_{s= 1}^m X_{H_s}(L \cap H_s)
\end{equation}

\end{itemize}

Let $\hat{G}_s^k \triangleq \{ i^{(k-1) \cdot (K-s+1) + 1}_s,...,i^{k \cdot (K-s+1)}_s\}$, where $s\in\{1,...,K\}, k\in\{1,...,\frac{K!}{K-s+1}\}$. And let $G^\ast_t \triangleq \{i^t_1,...,i^t_K \}$, where $t \in \{1,...,K!\}$. 
One can check that $\{ \hat{G}_s^k \}$ form a stable coalition structure, whereas $\{ G^\ast_t \}$ are a social optimum.
Note that if $H_s^t = \{i_s^t, ..., i_K^t\}$, then
\begin{equation}
p^{\rm ub}_{i_s^t}(H_s^t) \ge \frac{1}{2} 
\end{equation}
Therefore, the price of anarchy is lower bounded by
\begin{equation}
{\sf SPoA}^{\rm ub}_K \ge \frac{\sum_{t=1}^{K!} \sum_{s=1}^{K} p^{\rm ub}_{i_s^t}(H_s^t)}{\sum_{t=1}^{K!} c(H_1^t)} \ge \frac{K}{2}
\end{equation}
\end{proof}

\begin{customthm}{6} \label{thm:poa_ub1}
Consider usage based cost-sharing, such that ${\mathscr F}_i(\cdot)$ satisfies the monotone utilization condition. Then 
\begin{equation}
{\sf SPoA}^{\rm ub}_K \le \cH_K = \Theta(\log K).
\end{equation} 
\end{customthm}

\begin{proof}
Applying Lemma~\ref{lem:common} with $p_i=p_i^{\rm ub}$, we obtain 
$$
{\sf SPoA}^{\rm ub}_K \le\alpha\big(\{p_i^{\rm ub}(\cdot)\}_{i\in{\cal N}}\big)=\max_{c(\cdot),~H_1\supset\cdots\supset H_{K}}\frac{1}{c(H_1)}\sum_{s=1}^{K}p_{i_s}^{\rm ub}(H_s) 
$$
where $H_s = \{i_s, ..., i_K \}$, with the corresponding default costs denoted by $\{c_s, ..., c_K \}$.
Without loss of generality, we assume $c(H_1) = 1$.
Recall that $X_s(L) \triangleq X_{H_s}(L)$ and $p^{\rm ub}_{i_s}(H_s) = \sum_{L \subseteq H_s: i_s\in L} \frac{X_s(L)}{|L|}$.

Note that the monotone utilization condition is equivalent to saying that, for all $s\in\{1,...,K-1\}$ and $K\ge t>s$,
\begin{equation}
\sum_{L \subseteq H_t} X_t(L) \le \sum_{L \subseteq H_{s}: L\cap H_t \ne \varnothing} X_{s}(L)
\end{equation}
Hence, we can bound ${\sf SPoA}^{\rm ub}_K$ by the maximum value of the linear optimization problem (\textsc{P1}):
\begin{align}
& (\textsc{P1}) \ \max_{\{X_s(\cdot)\}_{s=1}^K} \sum_{s=1}^K\sum_{L\subseteq H_s:i_s\in L} \frac{X_s(L)}{|L|}\label{primal}\\
& \text{subject to} \notag\\
& \quad \sum_{L \subseteq H_t} X_t(L) \le \sum_{L \subseteq H_{s}:L\cap H_t \ne \varnothing} X_{s}(L), \notag\\
& \qquad \qquad \text{for all } 1\le s<t\le K, \label{b1}\\
&\quad \sum_{L\subseteq H_1}X_1(L)\le 1,~\text{for all } s\in\{1,...,K\}, \label{b2} \\
&\quad X_s(L)\ge 0,~\text{for all } s\in\{1,...,K\}, L\subseteq H_s  
\end{align}

For $s\in\{1,...,K\}$ and $L\subseteq H_s$, we define 
\begin{equation}
\rho(L,s)\triangleq\left\{
\begin{array}{ll}
\frac{1}{|L|}, &\text{ if }i_s\in L,\\
0, & \text{otherwise}
\end{array}
\right.  
\end{equation}
Then the dual problem to (\textsc{P1}) can be written as follows:
\begin{align}
& (\textsc{D1}) \ \min_{\lambda(\cdot), z} \quad z \label{dual}\\
& \text{subject to} \notag\\
& \quad \sum_{s=1}^{t-1} \lambda(s,t)- \sum_{t+1\le s\le K:L\cap H_s \ne \varnothing}\lambda(t,s)\ge \rho(L,t), \notag\\
& \qquad \qquad \text{for all }t\in\{2,...,K\},~L\subseteq H_t, \label{b3}\\
&\quad \rho(L,1)+ \sum_{2\le s\le K:L\cap H_s \ne \varnothing}\lambda(1,s)\le z,~\text{for all } L\subseteq H_1, \label{b4} \\
 &\quad \lambda(s,t)\ge 0,~\text{for all } 1\le s<t\le K,~z\ge 0
\end{align}

We next provide a primal-dual feasible pair $(X^\ast,\lambda^\ast)$ whose objective value is $\sum_{s=1}^K\frac{1}{s}$. To better understand this proof, it may be instructive to look at the example when $K=3$ in Table~\ref{tab:usage-k3}.

\medskip

\noindent
{\it Primal solution}:

For $s\in\{1,...,K\}$ and $L\subseteq H_s$, we set
\begin{equation}\label{ps}
X_s^\ast(L)\triangleq\left\{
\begin{array}{ll}
1, &\text{ if }L=H_s,\\
0, & \text{otherwise}
\end{array}
\right.   
\end{equation}
For $ s\in\{1,...,K-1\},~s<t\le K$, we obtain 
\begin{align}
1=&X^\ast_t(H_t)=\sum_{L \subseteq H_t} X_t^\ast(L)\\
=&X^\ast_s(H_s)=\sum_{L \subseteq H_{s}:L\cap H_t \ne \varnothing} X_{s}^\ast(L)
\end{align}
Also, we obtain 
\begin{equation}
\sum_{L\subseteq H_1}X_1^\ast(L)=X_1^\ast(H_1)=1
\end{equation}
Hence, $X^\ast$ satisfies Constraint~(\ref{b1}) and Constraint~(\ref{b2}). 

\medskip

\noindent
{\it Dual solution}: 

We first claim that there is a set of numbers $\lambda^\ast(s,t)\ge 0,$ for $1\le s<t\le K$ that satisfy Constraint~(\ref{b3}) as {\it equalities}. To show this, we study the linear optimization problem (\textsc{D2}):

\begin{align}
&(\textsc{D2}) \ \min_{\lambda(\cdot)} \sum_{1\le s<t\le K}\lambda(s,t) \label{dual2}\\
& \text{subject to} \notag\\
& \quad \sum_{s=1}^{t-1} \lambda(s,t)-\sum_{t+1\le s\le K:L\cap H_s \ne \varnothing} \lambda(t,s)=\rho(L,t), \notag\\
& \qquad \qquad \text{for all } t\in\{2,...,K\},~L\subseteq H_t, \label{b5}\\
& \quad \lambda(s,t)\ge 0,~\text{for all } 1\le s<t\le K
\end{align}
and its dual:
\begin{align}
& (\textsc{P2}) \ \max_{X_s(\cdot)} \sum_{s=2}^K\sum_{L\subseteq H_s:i_s\in L} \frac{X_s(L)}{|L|} \label{primal2} \\
& \text{subject to} \notag\\
& \quad \sum_{L \subseteq H_t} X_t(L) \le \sum_{L \subseteq H_{s}:L\cap H_t \ne \varnothing} X_{s}(L)+1,  \notag\\
& \qquad \qquad \text{for all } 1<s<t\le K, \label{b6}\\
& \quad \sum_{L \subseteq H_tu} X_t(L) \le 1, ~\text{for all } 1<t\le K \label{b7}
\end{align}

By Constraint~(\ref{b7}), the primal problem (\textsc{P2}) is bounded; it is also feasible since $X^\ast_s(L)$ for $s\in\{2,...,K\}$ given in Eqn.~(\ref{ps}) satisfies Constraint~(\ref{b6}) and Constraint~(\ref{b7}). It follows that the dual (\textsc{D2}) is also feasible and bounded, that is, there exist numbers $\lambda^\ast(s,t)$ satisfying Constraint~(\ref{b5}).

Let $z^\ast=\sum_{s=1}^K\frac{1}{s}$. We claim that $(\lambda^\ast,z^\ast)$ is a feasible solution to the dual problem Eqn.~(\ref{dual}). Evidently, we need only to check that it satisfies Constraint~(\ref{b4}).

For $L\subseteq H_1$, we sum the equations in Constraint~(\ref{b5}) for $L\cap H_t$ for all $t\in\{2,...,K\}$ to obtain
\begin{align}
&\sum_{t=2}^K\rho(L\cap H_t,t) \\
=&\sum_{t=2}^K\sum_{s=1}^{t-1} \lambda^\ast(s,t)-\sum_{t=2}^K\sum_{t+1\le s\le K:L\cap H_s \ne \varnothing}\lambda^\ast(t,s)\\
\ge&\sum_{t=2}^K\sum_{s=1}^{t-1} \lambda^\ast(s,t)-\sum_{t=2}^K\sum_{s=t+1}^K\lambda^\ast(t,s)\\
=&\sum_{s=1}^{K-1}\sum_{t=s+1}^{K} \lambda^\ast(s,t)-\sum_{t=2}^K\sum_{s=t+1}^K\lambda^\ast(t,s)\\
=&\sum_{s=2}^K \lambda^\ast(1,s)
\end{align}
It follows that 
\begin{align}
&\rho(L,1)+\sum_{s=2}^K \lambda^\ast(1,s)\le\sum_{t=1}^K\rho(L\cap H_t,t)\\
=&\sum_{1\le t\le K:i_t\in L}\frac{1}{|L\cap H_t|}=\sum_{t=1}^{|L|}\frac{1}{t}\le z^\ast
\end{align}
Hence, Constraint~(\ref{b4}) is satisfied.
\medskip

\noindent
{\it Optimality:} 

Finally, the proof is completed by noting that 
\begin{align}
\sum_{s=1}^K\sum_{L\subseteq H_s:i_s\in L} \frac{X^\ast_s(L)}{|L|}=\sum_{s=1}^K\frac{1}{|H_s|}=z^\ast
\end{align}
\end{proof}

\begin{table*}[htb!]
{\scriptsize
\begin{equation*}
\begin{array}{rl}
\hline\hline\\
(\textsc{P1}) \  & \max \ X_1(\{1\})+\frac{X_1(\{1,2\})}{2}+\frac{X_1(\{1,3\})}{2}+\frac{X_1(\{1,2,3\})}{3}+X_2(\{2\})+\frac{X_2(\{2,3\})}{2}+X_3(\{3\})\\
\text{subject to} &\\
{\lambda(1,2)}:&  X_2(\{2\})+X_2(\{3\})+X_2(\{2,3\}) \le  X_1(\{2\})+X_1(\{3\})+X_1(\{1,2\})+X_1(\{1,3\})+ X_1(\{2,3\})+X_1(\{1,2,3\})\\
{\lambda(1,3)}:&X_3(\{3\})\le X_1(\{3\})+X_1(\{1,3\}) + X_1(\{2,3\})+X_1(\{1,2,3\})  \\
{\lambda(2,3)}:&X_3(\{3\})\le X_2(\{3\})+ X_2(\{2,3\})  \\
{z}:& X_1(\{1\})+X_1(\{2\})+X_1(\{3\})+X_1(\{1,2\})+X_1(\{1,3\})+ X_1(\{2,3\})+X_1(\{1,2,3\})\le 1\\
&X_1(\{1\}),X_1(\{2\}),X_1(\{3\}),X_1(\{1,2\}),X_1(\{1,3\}),X_1(\{2,3\})\ge 0,\\
&X_1(\{1,2,3\}),X_2(\{2\}),X_2(\{3\}),X_2(\{2,3\}),X_3(\{3\})\ge  0 \\ 
\\\hline\\
(\textsc{D1}) \ &  \min \ z\\
\text{subject to} &\\
{X_2(\{2\})}: & \lambda(1,2) \ge  1 \\
{X_2(\{3\})}: & \lambda(1,2)-\lambda(2,3) \ge  0 \\
{X_2(\{2,3\})}:  & \lambda(1,2)-\lambda(2,3) \ge \frac{1}{2} \\
{X_3(\{3\})}: &\lambda(1,3)+\lambda(2,3) \ge  1\\
{X_1(\{1\})}: &1\le z\\
{X_1(\{2\})}: & \lambda(1,2)\le z\\
{X_1(\{3\})}: &\lambda(1,2) + \lambda(1,3)\le z\\
{X_1(\{1,2\})}:&\frac{1}{2}+  \lambda(1,2)\le z\\
{X_1(\{1,3\})}: &\frac{1}{2}+  \lambda(1,2)+\lambda(\{3\},1,3)\le z\\
{X_1(\{2,3\})}: & \lambda(1,2)+ \lambda(1,3)\le z\\
{X_1(\{1,2,3\})}: &\frac{1}{3}+ \lambda(1,2)+ \lambda(1,3)\le z\\
& \lambda(1,2),\lambda(1,3),\lambda(2,3)\ge 0 \\
\\\hline\\
(\textsc{D2}) \  & \min \ \lambda(1,2)+\lambda(1,3)+\lambda(2,3)\\
\text{subject to} &\\
{X_2(\{2\})}: & \lambda(1,2)=  1 \\
{X_2(\{3\})}: & \lambda(1,2)-\lambda(2,3)=  0 \\
{X_2(\{2,3\})}:  & \lambda(1,2)-\lambda(2,3) = \frac{1}{2} \\
{X_3(\{3\})}: &\lambda(1,3)+\lambda(2,3) =  1\\
& \lambda(1,2),\lambda(1,3),\lambda(2,3)\ge 0\\
\\\hline\\
(\textsc{P2}) \  &  \max \ X_2(\{2\})+\frac{X_2(\{2,3\})}{2}+X_3(\{3\})\\
\text{subject to} &\\
{\lambda(1,2)}:&  X_2(\{2\})+X_2(\{3\})+X_2(\{2,3\}) \le  1\\
{\lambda(1,3)}:&X_3(\{3\})\le 1 \\
{\lambda(2,3)}:&X_3(\{3\})\le X_2(\{3\})+ X_2(\{2,3\})+1 \\ \\
\hline\hline
\end{array}
\end{equation*}
}
\caption{(\textsc{P1}), (\textsc{D1}), (\textsc{D2}), (\textsc{P2}) for $K=3$.} \label{tab:usage-k3}
\end{table*}

\begin{customthm}{7} \label{thm:poa_ub2}
Consider usage based cost-sharing for pass sharing problem with uniform average cost $c^f = 1$. Then 
\begin{equation}
{\sf SPoA}^{\rm ub}_K \le \cH_K+1=\Theta(\log K).
\end{equation} 
\end{customthm}
\begin{proof}
Recall that ${\sf T}_i$ is the a set of required usage timeslots of user $i$, ${\sf T}_r$ is the set allowable timeslots of pass $r$, and ${\mathscr F}_i(r(G))= {\sf T}_i \cup \big({\sf T}_r\backslash(\bigcup_{j \in G} {\sf T}_j)\big)$. For pass sharing problem, we note that $X_G(L) = 0$ if $1 < |L| < |G|$. Hence,
\begin{equation}
p^{\rm ub}_{i}(G) = \frac{X_G(G)}{|G|} + X_G(\{i\})
\end{equation}
Since the average cost $c^f = 1$, $X_G(\{i\}) = |{\sf T}_i|$, and $X_G(G) = |{\sf T}_r\backslash(\bigcup_{j \in G} {\sf T}_j)|$.

Let $H_s \triangleq \{i_s, ..., i_K \}$ and $\hat{X} \triangleq \max_{s \in \{1,...,K\}} X_s(H_s)$.  Applying Lemma~\ref{lem:common} with $p_i=p_i^{\rm ub}$, we obtain 
\begin{equation}
{\sf SPoA}^{\rm ub}_K \le\max_{c(\cdot),~H_1\supset\cdots\supset H_{K}} \frac{1}{c(H_1)} \Big( \sum_{s=1}^{K} \big( \frac{\hat{X}}{|H_s|} + |{\sf T}_{i_s}| \big) \Big)
\end{equation}

Note that $\hat{X} \le c(H_1)$ because of monotonicity of cost function, and $\sum_{s=1}^{K} |{\sf T}_{i_s}| \le c(H_1)$, because the coalition of users $H_1= \{i_1, ..., i_K \}$ cannot overlap in their required usage timeslots, which are within the allowable timeslots of a pass utilized in $c(H_1)$. Therefore,
\begin{equation}
{\sf SPoA}^{\rm ub}_K \le 1 + \max_{c(\cdot)} \sum_{s= 1}^{K} \frac{1}{s} = \cH_K+1
\end{equation}
\end{proof}

\subsection{Existence of Stable Coalition Structures}  

This section investigates the existence of stable coalition structures considering different cost-sharing mechanisms.
First, we define a {\em cyclic preference} as sequences $(i_1,..., i_s)$ and $(G_1,..., G_s)$, where $i_k \in G_k\cap G_{k+1}$ for all $k \le s-1$, and $i_s \in G_s\cap G_1$, such that
\begin{eqnarray*}
u_{i_1}(p_{i_1}(G_1)) & > & u_{i_1}(p_{i_1}(G_2)), \\
u_{i_2}(p_{i_2}(G_2)) & > & u_{i_2}(p_{i_2}(G_3)), \\
& \vdots & \notag \\
u_{i_s}(p_{i_s}(G_s)) & > & u_{i_s}(p_{i_s}(G_1))
\end{eqnarray*}

\begin{customlemma}{11} \label{lem:cyc}
If there exists no cyclic preference, there always exists a stable coalition structure. Furthermore, such a stable coalition structure can be found in time $n^{O(K)}$.
\end{customlemma}

\begin{proof}
We include the standard argument for completeness (see, e.g., \cite{HVW15} for dynamic coalition formation by local improvements). Consider a directed graph $\cG=(\cN_K,E)$ on the set $\cN_K\triangleq\{S\in 2^\cN:|S|\le K\}$ of subsets of size at most $K$. For two sets $G_1,G_2\in\cN_K$, we define an edge $(G_1,G_2)\in E$ if and only if there is a participant $i\in G_1\cap G_2$ such that $u_i(G_1)<u_i(G_2)$. Then the existence of a cyclic preference is equivalent to the existence of a directed cycle in $\cG$. Thus if there exists no cyclic preference, then $\cG$ is acyclic and hence has at least one sink. 

Let $\cP$ be a maximal subset of sinks in $\cG$ with the property that any two distinct nodes $G,G'\in\cP$ are pairwise disjoint. 
Let $S$ be the set of participants covered by $\cP$, and $\cG'$ be the subgraph of $\cG$ obtained by deleting all the nodes containing some participant in $S$. 

By induction, there is a stable coalition structure $\cP'$ among the set of participants $S'\triangleq\cN\backslash S$. It follows that $\cP\cup\cP'$ is a stable coalition structure on the set of all participants. Indeed, if there is a blocking coalition $G_1$ then $G_1\cap S\neq\varnothing$ (since otherwise $\cP'$ is not stable among the participants in $S'$). But then there must exist $i\in S$ such that contain $u_i(G_1)>u_i(G_2)$, where $G_2\in\cP$ is the coalition containing $i$. This would imply that $(G_2,G_1)\in E$ contradicting that $G_2$ is a sink in $\cG$.  
\end{proof}

\begin{customthm}{8} 
For equal-split cost-sharing, there always exists a stable coalition structure.
\end{customthm}
\begin{proof}
If there exists a cyclic preference, defined by $(i_1,..., i_s)$ and $(G_1,..., G_s)$, then
\begin{eqnarray*}
u_{i_1}(p^{\rm eq}_{i_1}(G_1)) = c_{i_1} - \frac{c(G_1)}{|G_1|} & > & c_{i_1} - \frac{c(G_2)}{|G_2|} = u_{i_1}(p^{\rm eq}_{i_1}(G_2)), \\
u_{i_2}(p^{\rm eq}_{i_2}(G_2)) = c_{i_2} - \frac{c(G_2)}{|G_2|} & > & c_{i_2} - \frac{c(G_3)}{|G_3|} = u_{i_2}(p^{\rm eq}_{i_2}(G_3)), \\
& \vdots & \notag \\
u_{i_s}(p^{\rm eq}_{i_s}(G_s)) = c_{i_s} - \frac{c(G_s)}{|G_s|} & > & c_{i_s} - \frac{c(G_1)}{|G_1|} = u_{i_s}(p^{\rm eq}_{i_s}(G_1))
\end{eqnarray*}
Summing the above equations, one obtains a contradiction $ 0 > 0$. This completes the proof by Lemma~\ref{lem:cyc}.
\end{proof}

\begin{customthm}{9}
For proportional-split cost-sharing, there always exists a stable coalition structure.
\end{customthm}
\begin{proof}
If $u_{i}(p^{\rm pp}_{i}(G_1)) > u_{i}(p^{\rm pp}_{i}(G_2))$, then
\begin{equation}
c_i - \frac{c_i \cdot c(G_1)}{\sum_{j \in G_1} c_j} > c_i - \frac{c_i \cdot c(G_2)}{\sum_{j \in G_2} c_j} \quad \Rightarrow \quad 
\frac{ c(G_1)}{\sum_{j \in G_1} c_j} < \frac{ c(G_2)}{\sum_{j \in G_2} c_j}
\end{equation}
Thus, if there exists a cyclic preference, then
\begin{equation}
\frac{ c(G_1)}{\sum_{j \in G_1} c_j} < \frac{ c(G_2)}{\sum_{j \in G_2} c_j}<\cdots < \frac{ c(G_s)}{\sum_{j \in G_s} c_j}<\frac{ c(G_1)}{\sum_{j \in G_1} c_j}
\end{equation}
Summing the above equations, one obtains a contradiction $ 0 > 0$. This completes the proof by Lemma~\ref{lem:cyc}.
\end{proof}

\begin{customthm}{10} \label{thm:eqaexist}
For egalitarian cost-sharing, there always exists a stable coalition structure.
\end{customthm}
\begin{proof}
If there exists a cyclic preference, then
\begin{align}
& u_{i_1}(p^{\rm ega}_{i_1}(G_1)) = \frac{(\sum_{j \in G_1} c_j) - c(G_1)}{|G_1|} \notag  \\
> &  \frac{(\sum_{j \in G_2} c_j) - c(G_2)}{|G_2|} = u_{i_1}(p^{\rm ega}_{i_1}(G_2)), \\
& u_{i_2}(p^{\rm ega}_{i_2}(G_2)) = \frac{(\sum_{j \in G_2} c_j) - c(G_2)}{|G_2|}  \notag  \\
> & \frac{(\sum_{j \in G_3} c_j) - c(G_3)}{|G_3|} = u_{i_2}(p^{\rm ega}_{i_2}(G_3)), \\
& \qquad \vdots  \notag \\
& u_{i_s}(p^{\rm ega}_{i_s}(G_s)) = \frac{(\sum_{j \in G_s} c_j) - c(G_s)}{|G_s|} \notag \\
> & \frac{(\sum_{j \in G_1} c_j) - c(G_1)}{|G_1|} = u_{i_s}(p^{\rm ega}_{i_s}(G_1))
\end{align}
Summing the above equations, one obtains a contradiction $ 0 > 0$. This completes the proof by Lemma~\ref{lem:cyc}.
\end{proof}

\begin{customthm}{11}
For Nash bargaining solution, there always exists a stable coalition structure, irrespective of the constraint of non-negative payments.
\end{customthm}
\begin{proof}
First, if the constraint of non-negative payments is not considered, then the existence of a stable coalition structure follows from Corollary~\ref{cor:nashega} and Theorem~\ref{thm:eqaexist}.

Second, if the constraint of non-negative payments is considered, then $u_i(p^{\rm nash}_{i}(G) )\le c_i$. Note that there exists at least one participant $i \in G$ for any $G \subseteq {\cal N}$, such that $p^{\rm nash}_{i}(G) > 0$. Otherwise, $\sum_{i \in G} p^{\rm nash}_{i}(G) = c(G)= 0$ (which we may exclude without loss of generality). 
Suppose that there exists a cyclic preference, defined by $(i_1,..., i_s)$ and $(G_1,..., G_s)$. Let $H_t \subseteq G_t$ be the set of participants with positive payment in each $G_t$, that is, $p^{\rm nash}_{i}(G_t) > 0$ for all $i \in H_t$. By Lemma~\ref{lem:nash2}, it follows that
\begin{equation}
u_i(p^{\rm nash}_{i}(G_t)) = 
\left\{
\begin{array}{ll}
\frac{(\sum_{j \in H_t} c_j) - c(G_t)}{|H_t|}, & \mbox{\ if\ } i \in H_t \\
c_i, & \mbox{\ if\ } i \not\in H_t 
\end{array}
\right.
\end{equation}
By Lemma~\ref{lem:nash2}, if $i \not\in H_t$, $u_i(p^{\rm nash}_{i}(G_t)) = c_i \le \frac{(\sum_{j \in H_t} c_j) - c(G_t)}{|H_t|}$.
Hence, $u_i(p^{\rm nash}_{i}(G_t)) \le \frac{(\sum_{j \in H_t} c_j) - c(G_t)}{|H_t|}$ for all $i \in G_t$.

Note that if $u_i(p^{\rm nash}_{i}(G_t)) > u_i(p^{\rm nash}_{i}(G_{r'}))$, then $u_i(p^{\rm nash}_{i}(G_{r'})) \ne c_i$ because $u_i(p^{\rm nash}_{i}(G) )\le c_i$. If there exists a cyclic preference, then
\begin{align}
 & \frac{(\sum_{j \in H_1} c_j) - c(G_1)}{|H_1|} \ge u_{i_1}(p^{\rm nash}_{i_1}(G_1))   \notag \\
  > & u_{i_1}(p^{\rm nash}_{i_1}(G_2)) = \frac{(\sum_{j \in H_2} c_j) - c(G_2)}{|H_2|}, \\
& \frac{(\sum_{j \in H_2} c_j) - c(G_2)}{|H_2|} \ge u_{i_2}(p^{\rm nash}_{i_2}(G_2))   \notag\\
 > & u_{i_2}(p^{\rm nash}_{i_2}(G_3)) = \frac{(\sum_{j \in H_3} c_j) - c(G_3)}{|H_3|} , \\
& & \qquad  \vdots  \notag \\
& \frac{(\sum_{j \in H_s} c_j) - c(G_s)}{|H_s|} \ge u_{i_s}(p^{\rm nash}_{i_s}(G_s))  \notag \\
 > & u_{i_s}(p^{\rm nash}_{i_s}(G_1)) = \frac{(\sum_{j \in H_1} c_j) - c(G_1)}{|H_1|} 
\end{align}
Summing the above equations, one obtains a contradiction $ 0 > 0$. This completes the proof by Lemma~\ref{lem:cyc}.
\end{proof}

\subsection{Usage Based Cost-Sharing}

In general, usage based cost-sharing can induce cyclic preference, and hence, possibly the absence of a stable coalition structure. However, we can show the existence of a stable coalition structure in some special cases.

\medskip

\subsubsection{Pass Sharing}

For any $K\ge 2$, we can show that there always exists a stable coalition structure in the pass sharing problem, by ruling out any cyclic preference. Without loss of generality, we assume the average cost rate is 1 (i.e., $c^f = 1$). If participants $i\in G$ share a pass $r$, then $i$'s payment is given by
\begin{equation}
p_{i}^{\rm ub}(G) = |{\sf T}_i| + \frac{1}{|G|} |{\sf T}_r \backslash (\cup_{j\in G}{\sf T}_j)|
\end{equation}
If $i$ prefers to share in coalition $G$ with pass $r$ rather than on $G'$ with pass $r'$, then $p_{i}^{\rm ub}(G) < p_{i}^{\rm ub}(G')$, namely, 
\begin{equation}
 \frac{1}{|G|} |{\sf T}_r \backslash (\cup_{j\in G}{\sf T}_j)|< \frac{1}{|G'|} |{\sf T}_{r'} \backslash (\cup_{j\in G'}{\sf T}_j)|.
\end{equation}
If there exists a cyclic preference defined by $(i_1,..., i_s)$, $(G_1,...,G_s)$ and $(r_1,..., r_s)$, then
\begin{align}
& \frac{1}{|G_1|} |{\sf T}_{r_1} \backslash (\cup_{j\in G_1}{\sf T}_j)| < \frac{1}{|G_2|} |{\sf T}_{r_2} \backslash (\cup_{j\in G_2}{\sf T}_j)| \notag \\
< & \cdots \notag \\
< & \frac{1}{|G_s|} |{\sf T}_{r_s} \backslash (\cup_{j\in G_s}{\sf T}_j)|<\frac{1}{|G_1|} |{\sf T}_r \backslash (\cup_{j\in G_1}{\sf T}_j)|.
\end{align}
This generates a contradiction. Hence, there always exists a stable coalition structure.

\medskip

\subsubsection{Hotel Room Sharing}

When $K=2$, we can show that there always exists a stable coalition structure, by ruling out any cyclic preference. Let $\tau_i = t^{\sf out}_i - t^{\sf in}_i$ be the interval length required by participant $i$, and $\tau_{i,j}$ be the length of the overlapped interval, if participants $i, j$ share a room. Without loss of generality, we assume the room rate is 1. Then, $i$'s payment is given by
\begin{equation}
p_{i}^{\rm ub}(\{i,j\}) = (\tau_i - \tau_{i,j}) + \frac{1}{2} \tau_{i,j} = \tau_i - \frac{1}{2} \tau_{i,j}
\end{equation}
If $i$ prefers to share with $j$ rather than $k$, then $p_{i}^{\rm ub}(\{i,j\}) < p_{i}^{\rm ub}(\{i,k\})$, namely, $\tau_{i,j} > \tau_{i,k}$.
If there exists a cyclic preference $(i_1,..., i_s)$, then
\begin{equation}
\tau_{i_1,i_s} > \tau_{i_1,i_2} > ... > \tau_{i_{s-1},i_s} > \tau_{i_1,i_s}
\end{equation}
This generates a contradiction. Hence, there always exists a stable coalition structure.

\medskip

\subsubsection{Taxi-ride Sharing}

There exists an instance with no stable coalition structure even for $K=2$, as illustrated in Fig.~\ref{fig:taxi-loop}. Participant $i_{k}$ can share a ride with participant $i_{k-1}$ or participant $i_{k+1}$ (whereas participant $i_s$ can share with $i_{s-1}$ or participant $i_1$). Let the cost from $v^{\sf s}_{i_s}$ to $v^{\sf d}_{i_{k-1}}$ be $c(v^{\sf s}_{i_s}, v^{\sf d}_{i_{k-1}})$. Assume that $c(v^{\sf s}_{i_s}, v^{\sf d}_{i_{k-1}})$ is identical for all $k$, so are $c(v^{\sf s}_{i_s}, v^{\sf s}_{i_{k+1}})$, $c(v^{\sf d}_{i_s}, v^{\sf d}_{i_{k+1}})$ and $c(v^{\sf s}_{i_{k}}, v^{\sf d}_{i_{k}})$ for all $k$. Also, we assume that 
\begin{align}
& c(v^{\sf s}_{i_{k}}, v^{\sf d}_{i_{k}}) > \frac{1}{2} c(v_{i_s}^{\sf s},v_{i_{k-1}}^d) + c(v^{\sf d}_{i_{k-1}}, v^{\sf d}_{i_{k}}) \notag \\
> & c(v^{\sf s}_{i_{k}}, v^{\sf s}_{i_{k+1}}) + \frac{1}{2} c(v^{\sf s}_{i_{k+1}}, v^{\sf d}_{i_{k}}) \label{eqn:taxi-ex}
\end{align}
Hence, participant $i_{k}$ prefers to share with participant $i_{k+1}$, rather than with participant $i_{k-1}$. This generates a cyclic preference $(i_1,..., i_s)$. If there are odd number of participants arranged in a loop, then this can give no stable coalition structure. We remark that Eqn.~(\ref{eqn:taxi-ex}) can be attained, when $s$ is sufficiently large.

\begin{figure}[!htb]
\center
\includegraphics[width=0.3	\textwidth]{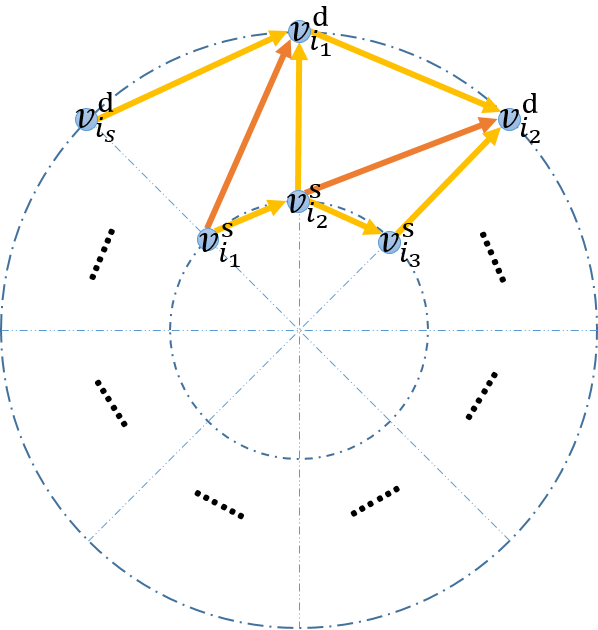}
\caption{An illustration for taxi-ride sharing problem with no stable coalition structure considering usage based cost-sharing.}
\label{fig:taxi-loop}
\end{figure}

\subsection{NP-Hardness} \label{sec:NP}

This section studies the hardness of solving \textsc{$K$-MinCoalition}. 

\begin{customthm}{12} \label{thm:nph}
\textsc{$K$-MinCoalition} is NP-hard for $K\ge 3$.
\end{customthm}
\begin{proof}
First, define the set of coalitions with at most size $K$ by $\cN_K\triangleq\{S\in 2^\cN:|S|\le K\}$.

\textsc{$K$-MinCoalition} can be reduced from NP-hard problem (\textsc{Exact-Cover-By-$K$-Sets}), defined as follows. Given a collection ${\cal G}$ of subsets of $\cN\triangleq\{1,...,n\}$, each of size $K$, find a pairwise-disjoint sub-collection ${\cal G}'$ that covers $\cN$, that is, $S\cap S'=\varnothing$ for all distinct $S,S'\in{\cal G}'$ and $\bigcup_{S\in{\cal G}'}=\cN$. 

Given an instance ${\cal G}$ of \textsc{Exact-Cover-By-$K$-Sets}, we construct an instance of \textsc{$K$-MinCoalition} as follows. For $S\in\cN_K$, we define $c(S)$ as follows:
\begin{equation}
c(S)=\left\{
\begin{array}{ll}
1&\text{ if }S\in{\cal G},\\
1 &\text{ if }S\in \cN_K\backslash {\cal G},~|S|\le K-1,\\
2 &\text{ if }S\in\cN_K\backslash {\cal G},~|S|= K
\end{array}
\right.
\end{equation}
One can check that (C1) and (C2) are satisfied. 
Now consider a feasible solution to \textsc{$K$-MinCoalition} and assume it consists of $n_i$ sets of size $i$ not from ${\cal G}$, for $i=1,...,K$, and $n_K'$ sets of size $K$ from ${\cal G}$. Then the total cost of the solution is $\sum_{i=1}^{K-1}n_i+n_K'+2n_K$. Subject to $\sum_{i=1}^{K-1}i\cdot n_i+K\cdot n_K'+K\cdot n_K=n$, this cost is uniquely minimized when $n_K'=n/K$, that is, when there is a disjoint collection from ${\cal G}$ covering $\cN$. Indeed, the minimum of linear relaxation of this integer programming problem is determined by the minimum ratio test: 
\begin{equation}
\min\left\{\min_{i\in\{1,...,K-1\}}\left\{\frac{1}{i}\right\},\frac{2}{K},\frac{1}{K}\right\}=\frac{1}{K}
\end{equation}
On the other hand, if the answer to the instance ${\cal G}$ of \textsc{Exact-Cover-By-$K$-Sets} is NO, then this unique minimum cannot be achieved by an integral solution, yielding a solution of cost strictly larger than $\frac{n}{K}$.  
\end{proof}

\end{document}